%% file: main.tex
\documentclass{llncs}

\usepackage{amssymb}
\usepackage{amsmath}
\usepackage{amscd}
\usepackage{amsxtra}
\usepackage{latexsym}
\usepackage{epic}
\usepackage{eepic}
\usepackage{psfrag}
\usepackage{psfig}
\usepackage{graphicx}
\usepackage[vlined,linesnumbered,algoruled]{algorithm2e}
\usepackage{luca,macro}
\usepackage{subfig,caption} 
\usepackage{wrapfig}

\usepackage[all]{xy}
\usepackage{fancyvrb}
\usepackage{float}
\usepackage{paralist}
\usepackage{color}
\usepackage{colortbl}

\begin{document}

\title{Assume-Guarantee Synthesis for Digital Contract Signing}

\author{
  Krishnendu Chatterjee \inst{1} \and
  Vishwanath Raman \inst{2}
}
\titlerunning{Assume Guarantee Synthesis for Digital Contract Signing}
\authorrunning{Chatterjee, Raman}

\institute{
  IST Austria (Institute of Science and Technology Austria)\\
  \email{krishnendu.chatterjee@ist.ac.at}
  \and
  Carnegie Mellon University, Moffett Field, USA\\
  \email{vishwa.raman@sv.cmu.edu}
}

\date{}

\maketitle

\begin{abstract}
\input{abstract}
\end{abstract}

\input{introduction}
\input{fairnonrep}

\input{specifications}
\input{definitions}

\input{cosynthesis}
\input{analysis}

\input{symmetric}
\input{conclusion}

\bibliographystyle{plain}
\bibliography{dvlab}

\input{appendix}

\end{document}

%% file: abstract.tex
We study the automatic synthesis of fair non-repudiation 
protocols, a class of fair exchange protocols, used for digital
contract signing.
First, we show how to specify the objectives of the participating 
agents and the trusted third party ($\ttp$) as path formulas in LTL
and prove that the satisfaction of these objectives imply 
\emph{fairness}; a property required of fair exchange protocols.
We then show that \emph{weak (co-operative) co-synthesis} and 
\emph{classical (strictly competitive) co-synthesis} fail, 
whereas \emph{assume-guarantee synthesis (AGS)} succeeds. 
We demonstrate the success of assume-guarantee synthesis as follows:
(a)~any solution of assume-guarantee synthesis is~\emph{attack-free}; 
no subset of participants can violate the objectives of the
other participants;
(b)~the Asokan-Shoup-Waidner (ASW) certified mail protocol 
that has known vulnerabilities is not a solution of AGS;
(c)~the Kremer-Markowitch (KM) non-repudiation protocol
is a solution of AGS; and 
(d)~AGS presents a new and symmetric fair non-repudiation 
protocol that is attack-free.
To our knowledge this is the first application of synthesis to fair 
non-repudiation protocols, and our results show how synthesis can 
both automatically discover vulnerabilities in protocols and 
generate correct protocols.
The solution to assume-guarantee synthesis can be computed
efficiently as the secure equilibrium solution of three-player
graph games.

%% file: introduction.tex
\section{Introduction}

\noindent{\bf Digital contract signing.} 
The traditional paper-based contract signing mechanism involves two 
participants with an intent to sign a piece of contractual text, that 
is typically in front of them.
In this case, either both of them agree and sign the contract or
they do not.
The mechanism is ``fair'' to both participants in that it does
not afford either participant an unfair ``advantage'' over the 
other.
In digital contract signing, ubiquitous in the internet era,
an {\em originator\/} sends her intent to sign a contractual text
to a {\em recipient\/}.
Over the course of a set of messages, they then proceed to exchange
their actual signatures on the contract.
In this case, it is in general difficult to ensure fairness
as one of the two participants gains an advantage over the
other, during the course of the exchange.
If the participants do not trust each other, then neither 
wants to sign the contract first and the one that signs it first
may never get a reciprocal signature from the other participant.
Moreover, as these contracts are typically signed over asynchronous 
networks, the communication channels may provide no guarantees 
on message delivery.
The same situation arises in other related areas, such as fair 
exchange and certified email.

\noindent{\bf Protocols for digital contract signing.}
Many protocols have been designed to facilitate the exchange of 
digital signatures.
The earliest exchange protocols were probabilistic.
Participants transmit successive bits of information, under
the expectation that both participants have similar computation
power to detect dishonest behavior and stop participating 
in the protocol.
These protocols are impractical as the number of messages
exchanged may be very large, and both participants having similar 
computation power may not be realistic.
Even and Yacobi \cite{Yacobi80} first showed that no deterministic 
contract signing protocol can be realized without the involvement 
of a third party arbitrator who is trusted by all participants.
This was formalized as an impossibility result in \cite{Pagnia99},
where the authors show that fair exchange is impossible without 
a {\em trusted third party ($\ttp$)\/} for non-repudiation protocols.
A simple protocol with a $\ttp$ has a $\ttp$ collect all 
signatures and then distribute them to the participants.
But this is inefficient as it involves an online $\ttp$ to facilitate
every exchange, easily creating a bottleneck at the site of the
$\ttp$.
This has lead to the development of {\em optimistic protocols\/},
where two participants exchange their signatures without 
involving a $\ttp$, calling upon the $\ttp$ to adjudicate only
when one of the two participants is dishonest.
These protocols are called {\em fair non-repudiation protocols\/}
with {\em offline\/} $\ttp$.

\medskip\noindent{\bf Fair non-repudiation protocols.} 
A \emph{fair non-repudiation} protocol is therefore a contract 
signing protocol, falling under the category of fair exchange
protocols, that ensures that at the end of the 
exchange of signatures over a network, neither participant can
deny having participated in the protocol.
A non-repudiation protocol, upon successful termination, provides
each participant evidence of commitment to a contract
that cannot be repudiated by the other participant.
A {\em non-repudiation of origin (NRO)\/} provides the recipient
in an exchange, the ability to present to an adjudicator, 
evidence of the senders commitment to a contract.
A {\em non-repudiation of receipt (NRR)\/} provides the sender in an
exchange, the ability to present to an adjudicator, evidence of the
recipient's commitment to a contract.
An exchange protocol should satisfy the following informal
requirements \cite{MarkowitchGK02,GarayJM99}:
\begin{enumerate}
\item {\em Fairness.} The communication channels quality being 
fixed, at the end of the exchange protocol run, either all
involved parties obtain their expected items or none (even a part)
of the information to be exchanged with respect to the missing items
is received.
\item {\em Abuse-freeness.} It is impossible for a single entity
at any point in the protocol to be able to prove to an outside
party that she has the power to terminate (abort) or successfully
complete the protocol.
\item {\em Timeliness.} The communication channels quality being
fixed, the parties always have the ability to reach, in a finite
amount of time, a point in the protocol where they can stop the
protocol while preserving fairness.
\end{enumerate}

In addition to the above properties, a fair non-repudiation
protocol is also expected to satisfy the following requirements:
(a) {\em Viability.\/} Independently of the communication channels
quality, there exists an execution of the protocol, where the
exchange succeeds.
(b) {\em Non-repudiability.\/} It is impossible for a single
entity, after the execution of the protocol, to deny having
participated in a part or the whole of the communication.

\medskip\noindent{\bf Existing protocols.}
Some of the existing protocols in this category are the 
Zhou-Gollmann (ZG) protocol \cite{Zhou97}, the Asokan-Shoup-Waidner 
(ASW) protocol \cite{AsokanSW98}, the 
Garay-Jakobsson-MacKenzie (GJM) protocol \cite{GarayJM99} and
the Kremer-Markowitch (KM) protocol \cite{MarkowitchK01}.
Non-repudiation protocols are difficult to design in general 
\cite{ZhouDB00,ShmatikovM02,MarkowitchGK02,KremerMZ02,KremerR03} and
much literature covers the design and verification of
these protocols.
While some of the literature covers the discovery of 
vulnerabilities in these protocols based on the content of the 
exchanged messages, others have tried to find attacks based on 
the sequences of messages that can be exchanged, based on the rules
of the protocols.
However, there is no work that focuses on automatically obtaining 
correct solutions of these subtle and hard to design protocols.

\medskip\noindent{\bf Our contributions.} 
We show that the classical synthesis formulations that are strictly
competitive are inadequate for synthesizing these protocols and
that newer {\em conditionally competitive\/} formulations are more 
appropriate.
To our knowledge this is the first application of game-theoretic
controller synthesis to security protocols.
Synthesis has many advantages over model checking.
While model checking finds specific vulnerabilities for a designed
protocol, the counter-examples in synthesis are strategies (or
refinements) that exhibit vulnerabilities against a set of
protocol realizations.
Moreover, impossibility results such as failure to realize non-repudiation
protocols without a TTP cannot be deduced with model checking, whereas
such results can be deduced in a synthesis framework, as we show in
this paper.
Our contributions are as follows:
\begin{enumerate}
\item
We present the formal objectives of the participants
and the trusted third party as path formulas in Linear Temporal
Logic (LTL) and prove that satisfaction of the objectives imply
fairness of the protocols
(for syntax, semantics and a description of LTL see \cite{Pn1,MPvol1}).
The timeliness property is also satisfied easily.
The precise formalization of protocol requirements as LTL path
formulas is a basic pre-requisite for synthesis.

\item
We show that classical (strictly competitive) co-synthesis and
weak (co-operative) co-synthesis fail, whereas assume-guarantee
(conditionally competitive) co-synthesis \cite{CH07} succeeds.

\item We show that all solutions in the set $P_{AGS}$ of assume-guarantee
solutions are \emph{attack-free}, i.e., any solution in $P_{AGS}$ 
prevents malicious participants from gaining an unfair advantage.

\item
We show that the ASW certified mail protocol is not in $P_{AGS}$, 
due to known vulnerabilities that could have been automatically 
discovered.
The GJM protocol is also not in $P_{AGS}$ as it compromises
our objective for the $\ttp$, while providing fairness and abuse-freeness 
to the agents.
The KM protocol is in $P_{AGS}$ and it follows that it could have been 
automatically generated by formalizing the problem of protocol 
design as a synthesis problem.

\item The ASW, GJM and the KM protocol are not symmetric as they do not 
allow the recipient to abort the protocol. 
From our analysis of the refinements in $P_{AGS}$ we construct a \emph{new} 
and \emph{symmetric} fair non-repudiation protocol that provides not just 
the originator but also the recipient in an exchange, the ability to 
abort the protocol.
Given that the $\ttp$ satisfies certain constraints on her behavior,
such that her objective is satisfied, we show that the symmetric protocol
is attack-free.

\item
Our results provide a game-theoretic justification of the need for a
trusted third party. 
This gives an alternative justification of the 
impossibility results of \cite{Yacobi80,Pagnia99}.
\end{enumerate}
It was shown in~\cite{CH07} that the solutions of assume-guarantee 
synthesis can be obtained through the solution of secure equilibria
\cite{CHJ06} in graph games.
Applying the results of~\cite{CH07}, given our objectives, we show 
that for fair non-repudiation protocols, 
the solutions can be obtained in quadratic time.

\medskip\noindent{\bf Related works.}
The formal verification of fair exchange protocols uses model checking
to verify a set of protocol objectives specified in a suitable temporal
logic.
The work of Shmatikov and Mitchell \cite{ShmatikovM02} uses the
finite state tool Mur$\varphi$ to model the participants in a protocol
together with an intruder model, to check a set of safety properties
by state space exploration.
They expose a number of vulnerabilities that may lead to replay
attacks in both the ASW protocol and the GJM protocol.
Zhou et al., show the use of belief logics to verify non-repudiation
protocols \cite{Zhou98}.
The works \cite{Kremer02gameanalysis,KremerMZ02,KremerR03,ChadhaKS06} use
game theoretic models and the logic ATL to formally specify fairness,
abuse-freeness and timeliness, that they verify using the tool
MOCHA \cite{MOCHA}.
Independently, in \cite{ChadhaKS01} the authors use a game-based approach,
with a set-rewriting technique, to verify fair exchange protocols.
However, these works focus on verification and not on the synthesis of
protocols.
Armando et al., \cite{ArmandoCC07} use set-rewriting
with LTL to verify the ASW protocol and report a new
attack on the protocol.
Louridas in \cite{Louridas} provides several insightful guidelines for
the design of non-repudiation protocols.

The notion of weak or {\em co-operative\/} co-synthesis was introduced
in \cite{CE81}, classical or {\em strictly competitive\/} co-synthesis
was studied in \cite{PnueliRosner89,RW87} and
assume-guarantee or {\em conditionally competitive\/} co-synthesis
was introduced in \cite{CH07}.
But none of these works consider security protocols.
The first effort at synthesizing security protocols is
\cite{PerrigS00,Song01athena} and is related to the automatic
generation of mutual authentication protocols, where the authors use
iterative deepening with a cost function to generate
correct protocols that minimize cost; they do not address
digital contract signing.
In \cite{Saidi02}, the authors describe a prototype synthesis tool
that uses the BAN \cite{BurrowsAN90} logic to describe protocol goals with
extensions to describe protocol rules that, when combined with a proof
system, can be used to generate protocols satisfying those goals.
The authors use their approach to synthesize the Needham-Schroeder
protocol; they do not address digital contract signing.
The work of \cite{AlcaideECR08} uses multi-player games to obtain
correct solutions of multi-party rational exchange protocols in the
emerging area of rational cryptography.
These protocols do not provide fairness, but do ensure that rational 
parties would have no reason to deviate from the protocol.
None of the above works use a conditionally competitive synthesis
formulation, which we show is necessary for fair non-repudiation
protocols.
Our technique is very different from these and all previous works,
as we use the rich body of research in controller synthesis to
construct fair exchange protocols efficiently; in time that is
quadratic in the size of the model.
The finite state models are typically small, so that the application
of synthesis techniques as we propose in this paper is both appealing
and realizable in practice.

%% file: fairnonrep.tex
\section{Fair Non-repudiation Protocols}
\label{sec:fair-nonrep}

In this section we introduce fair non-repudiation protocols.
We first define a participant model, a protocol model and an attack model.
We then introduce the agents and the trusted third party that
participate in fair exchange protocols, the messages that they 
may send and receive, and the channels over which they 
communicate.
Finally, we introduce a set of predicates that are
set based on messages that are sent and received and that form the
basis for our protocol and participant objectives in the subsequent
section.

\medskip\noindent{\bf A participant model.}
Our protocol model is different from the Strand Space model 
\cite{ThayerHG99} and is closer to the model required for the synthesis 
of protocols as participant refinements.
We define our model as follows:
Let $V$ be a finite set of variables that take values in some
domain $D_v$.
A {\em valuation\/} $f$ over the variables $V$ is a function
$f: V \mapsto D_v$ that assigns to each variable $v \in V$, a
value $f(v) \in D_v$; we take $\calf[V]$ as the set of all
valuations over the variables in $V$.
Let $\calm$ be a finite set of messages ({\em terms\/} in the 
Strand Space model) that are exchanged between a
set $A = \set{A_i \ \vert \ 0 \le i \le n}$ ({\em roles\/} in the
Strand Space model) of participants.
We define each participant as the tuple
$A_i = (L_i, V_i, \mas_i, \trans_i)$, where $L_i$ is a finite set of
control points or values taken by a program counter, $V_i \subseteq V$ is 
a set of variables, 
$\mas_i: \calf[V_i] \mapsto 2^\calm$ is a message assignment, that 
given a valuation $f \in \calf[V_i]$, returns the set 
of messages that can be sent by $A_i$ at $f$; this set
includes all messages that can be composed by $A_i$ based on what
she knows in the valuation $f$.
Valuations over variables represent what a participant
knows at a given control point.
We take $V = \bigcup_{i = 0}^n V_i$ and assume that the sets $V_i$ 
form a partition of $V$.
An $A_i$ transition function is $\trans_i: L_i \times \calf[V_i] \times \calm 
\mapsto L_i \times \calf[V_i]$, that given a control point, 
a valuation over $V_i$ and a message either sent or received by $A_i$, 
returns the next control point of $A_i$ and an updated valuation.
The participants may send messages simultaneously and independently,
and can either receive a message or send a message at every control
point.

\medskip\noindent{\bf The most general participants.}
We interpret the elements of $A$ as the {\em most general participants\/}
in an exchange; the participants in $A$ can send any message that can be 
composed at each control point, based on messages they have received up to 
that control point.
We take the interaction between the elements of $A$ as the 
{\em most general exchange program\/}.
Every participant in an exchange has her own objective to 
satisfy.
We take the objective of a participant as a set of desired
sequences of valuations of the protocol variables.

\medskip\noindent{\bf A protocol model.}
A realization of an exchange protocol is a restriction of the most 
general exchange program that consists of the set 
$A' = \set{A_i' \ \vert \ 0 \le i \le n}$ of participants, with 
behaviors restricted by the rules of the protocol.
We take $A_i' = (L_i', V_i, \mas_i', \trans_i')$, where
$L_i' \subseteq L_i$; $V_i$ is the same set of variables as in $A_i$;
for every valuation $f \in \calf[V_i]$ we have
$\mas_i'(f) \subseteq \mas_i(f)$;
and $\trans_i': L_i' \times \calf[V_i] \times \calm \mapsto 
L_i' \times \calf[V_i]$ is the transition function, that given a
control point in $L_i'$, a valuation over $V_i$ and a message either 
sent or received by $A_i'$ returns the 
next control point of $A_i'$ and an updated valuation.
For $l \in L_i'$, $v \in \calf[V_i]$ and $m \in \calm$, we have 
$\trans_i'(l, v, m) = \trans_i(l, v, m)$.
We define a {\em protocol instance\/} (or a {\em protocol run\/}) 
as any sequence of valuations generated by the participants in $A'$
and take the set of all possible protocol runs as {\em Runs\/}$(A')$.
We refer to a message that can be sent by a participant as a
{\em move\/} of that participant.

\medskip\noindent{\bf An attack model.}
We define an {\em attack\/} on a protocol as the behavior of
a subset of protocol participants such that the resulting sequence
of messages is in their objective but not in the objective
of at least one of the other participants.
Formally, let $Y \subseteq A$ be a subset of the most general 
participants with $(A \setm Y)' = \set{A_i' \vert A_i \in (A \setm Y)}$ 
being the remaining participants that follow the rules of the protocol.
A protocol has a $Y$-attack if the most general participants in $Y$ 
can generate a message sequence, given $(A \setm Y)'$ follow the
protocol, that is not in the objective of at least one participant
in $(A \setm Y)'$ but is in the objectives of all participants in $Y$.
A protocol is {\em attack-free\/}, if there exists no $Y$-attack
for all $Y \in 2^A$.

\medskip\noindent{\bf Agents.}
An {\em agent\/} in a two-party exchange protocol is one of the two 
participating entities signing an online contract.
Based on whether an agent proposes a contract or accepts a contract
originating from another agent, we get two roles that an agent can
play; that of an {\em originator\/} of a contract, designated by
$\O$ or the {\em recipient\/} of a contract, designated
by $\R$.
Agents communicate with each other over channels.

\medskip\noindent{\bf Trusted third party ($\ttp$).}
The {\em trusted third party\/} or $\ttp$ is a participant who is 
trusted by the agents and adjudicates and resolves disputes.
It is known that a fair exchange protocol cannot be realized without
the $\ttp$ \cite{Yacobi80,Pagnia99}.
We model the $\ttp$ explicitly as a participant, define her objective 
and using our formulation give a game-theoretic justification that
the $\ttp$ is necessary.
Agents and the $\ttp$ communicate with each other over channels.

\medskip\noindent{\bf Messages.}
A {\em message\/} is an encrypted stream of bytes; we treat each
message as an atomic unit.
We assume each message contains a {\em nonce\/} that uniquely
identifies a protocol instance; participants can simultaneously
participate in multiple protocol instances.
We are not concerned with the exact contents of each message, but in
what each message conveys; this is in keeping with our objective
of synthesizing protocols that are attack-free with respect to
message interleavings.
From the definition of messages in fair exchange protocols
in \cite{Kremer02gameanalysis,KremerMZ02,KremerR03,ShmatikovM02}
and other works, we define the set $\calm$ of messages as follows:
\begin{itemize}
\item $m_1$ is a message that may be sent by $\O$ to $\R$. 
The intent of this message is to convey $\O$'s desire
to sign a contract with a recipient $\R$.
\item $m_2$ is a message that may be sent by $\R$ upon
receiving $m_1$ to $\O$.
This conveys $\R$'s intent to sign the contract sent by $\O$.
\item $m_3$ is a message that may be sent by $\O$ to $\R$ upon
receiving $m_2$ and contains the actual signature of $\O$.
\item $m_4$ is a message that contains the actual signature of
$\R$ and may be sent by $\R$ to $\O$ upon receiving $m_3$.
\item $a_1^\O$ is a message that may be sent by $\O$ to the $\ttp$ 
and conveys $\O$'s desire to {\em abort\/} the protocol.
\item $a_2^\O$ (resp. $a_2^\R$) is a message that may be sent by the 
$\ttp$ to $\O$ (resp. $\R$) that confirms the abort by including an 
abort token for $\O$ (resp. $\R$).
\item $r_1^\O$ (resp. $r_1^\R$) is a message that may be sent by $\O$
(resp. $\R$) to the $\ttp$ and conveys $\O$'s (resp. $\R$'s) desire to 
get the $\ttp$ to {\em resolve\/} a protocol instance by explicitly 
requesting the $\ttp$ to adjudicate.
We do not specify the content of $r_1^\O$ or $r_1^\R$ but make the 
assumption that the $\ttp$ needs $m_1$ to recover the 
protocol for $\R$ and similarly needs $m_2$ to recover the 
protocol for $\O$.
\item $r_2^\O$ (resp. $r_2^\R$) is a message that may be sent by the 
$\ttp$ to $\O$ (resp. $\R$) and contains a universally verifiable 
signature in lieu of the signature of $\R$ (resp. $\O$).
\end{itemize}
The messages that each participant can send in a state depends
on what the participant {\em knows\/} in that state.
We assume that every recipient can check if the message she receives 
contains what she expects and that it originates from the purported 
sender.
We impose an order on the messages $m_1, m_2, m_3$ and $m_4$ as 
it can be shown trivially in our synthesis formulation that $\O$ sending 
$m_3$ before receiving $m_2$ and $\R$ sending $m_4$ before receiving 
$m_3$ violates their respective objectives.
Further, since our concern in this paper is not to synthesize messages 
impervious to attacks, but rather the correct sequences of messages that are
impervious to attacks, we assume the former can be accomplished by 
the use of appropriate cryptographic primitives.
We remark that primitives such as {\em private contract signatures 
(PCS)\/} introduced by Garay et al., in \cite{GarayJM99}, can be used 
with protocols that are synthesized using our technique to ensure
such properties as the {\em designated verifier property\/} which
guarantees abuse-freeness.
In our formulations, we consider a {\em reasonable $\ttp$\/} that 
satisfies the following restrictions on behavior:
\begin{enumerate}
\item The $\ttp$ will never send a message unless it receives
an abort or a resolve request.
\item The $\ttp$ processes messages in a 
first-in-first-out fashion.
\item If the first message received by the $\ttp$
is an abort request, then the $\ttp$ will eventually send an abort
token.
\item If the first message received by the $\ttp$ is a
resolve request, then the $\ttp$ will eventually send an agent 
signature.
\end{enumerate}

\medskip\noindent{\bf Channels.}
A channel is used to deliver a {\em message\/}.
There are three types of channels that are typically modeled in
the literature.
We present them here in decreasing order of reliability:
\begin{enumerate}
\item An {\em operational\/} channel delivers all messages within a 
known, finite amount of time.
\item A {\em resilient\/} channel eventually delivers all messages,
but there is no fixed finite bound on the time to deliver a message.
\item An {\em unreliable\/} channel may not deliver all messages
eventually.
\end{enumerate}
We model the channels between the agents as unreliable and those 
between the agents and the $\ttp$ as resilient as in prevailing
models;
messages sent to the $\ttp$ and by the $\ttp$ will 
eventually be delivered. 
We do not model the channels explicitly, but synthesize protocols 
irrespective of channel behavior.
In particular, unreliable channels may never deliver messages and
messages sent to the $\ttp$ may arrive in any order at the $\ttp$.

\medskip\noindent{\bf Scheduler.}
The scheduler is not explicitly part of any fair exchange protocol.
The protocol needs to provide all agents the ability to send
messages asynchronously.
This implies that the agents can choose their actions simultaneously
and independently.
We model this behavior by using a fair scheduler that assigns each 
participant a turn and we synthesize refinements
against all possible behaviors of a fair scheduler.

\medskip\noindent{\bf Predicates.}
We introduce the following set of predicates.
\begin{itemize}
\item $M_1$ is set by $\O$, when she sends message $m_1$ to
$\R$.
\item $\eoo$, referred to as the {\em Evidence Of Origin\/},
is set by $\R$ when either $m_1$ or $r_2^\R$ is received.
\item $\eor$, referred to as the {\em Evidence of Receipt\/},
is set by $\O$ when either $m_2$ or $r_2^\O$ is received.
\item $\siga^\O$ and $\siga^\ttp$ are referred to as 
{\em $\O$'s signature\/}.
$\siga^\O$ is set by $\R$ when $\R$ receives $m_3$ and $\siga^\ttp$
is set by $\R$ when he receives $r_2^\R$.
\item $\sigb^\R$ and $\sigb^\ttp$ are referred to as 
{\em $\R$'s signature\/}.
$\sigb^\R$ is set by $\O$ when $\O$ receives $m_4$ and $\sigb^\ttp$
is set by $\O$ when she receives $r_2^\O$.
\item $\abo$ is set by $\O$ and indicates that $a_2^\O$ has been
received.
\item $\abr$ is set by $\R$ and indicates that $a_2^\R$ has been
received.
\item $\reqa$ is set by the $\ttp$ when an abort request,
$a_1^\O$ is received.
\item $\reqr$ is set by the $\ttp$ when a resolve request, 
$r_1^\O$ or $r_1^\R$, is received.
\end{itemize}
All predicates are {\em monotonic\/} in that once they are set, 
they remain set for the duration of a protocol instance
\cite{ShmatikovM02}.
We distinguish between a signature sent by an agent and 
the signature sent by the $\ttp$ as a replacement for an agent's 
signature in the predicates.
Distinguishing these signatures enables modeling $\ttp$ 
accountability \cite{ShmatikovM02}.
The non-repudiation of origin for $\R$, denoted by $\nro$, means 
that $\R$ has received both $\O$'s intent to sign a contract and 
$\O$'s signature on the contract so that $\O$ cannot deny having 
signed the contract to a third party.
Formally, $\nro$ is defined as: $\nro = \eoo \wedge (\siga^\O \vee
\siga^\ttp)$.
The non-repudiation of receipt for $\O$, denoted by $\nrr$, means 
that $\O$ has received both the intent and signature of $\R$ on a 
contract so that $\R$ cannot deny having signed the contract to a 
third party.
Formally, $\nrr$ is defined as: $\nrr = \eor \wedge (\sigb^\R \vee
\sigb^\ttp)$.

%% file: specifications.tex
\section{LTL Specifications for Protocol Requirements}
\label{sec:specifications}

The synthesis of programs requires a formal objective of their
requirements.
One of our contributions in this paper is to present a precise and
formal description of the protocol requirement as a path formula in
Linear Temporal Logic (LTL \cite{Pn1,MPvol1}), which then becomes 
our synthesis objective.
In this section, we define the objective for fair non-repudiation
protocols, objectives for the agents and the $\ttp$ and show that 
satisfaction of the objectives of the agents and the $\ttp$ imply 
satisfaction of the objective of the protocols.
We use LTL, a logic that is used to specify properties of infinite
paths in finite-state transition systems.
In our specifications, we use the usual LTL notations $\bo$ 
and $\diam$ to denote \emph{always} (safety) and 
\emph{eventually} (reachability) specifications, respectively.

\medskip\noindent{\bf Fairness.\/}
Informally, fairness for $\O$ can be stated as {\em ``For all
protocol instances if the non-repudiation of origin ($\nro$) is
ever true, then eventually the non-repudiation of receipt ($\nrr$)
is also true''\/} \cite{KremerR03}.
The fairness property for $\O$ is expressed by the LTL formula
\begin{align*}
\varphi_f^\O = \bo (\nro \im \diam \nrr) \eqpun . 
\end{align*}
Similarly, the fairness property for $\R$ is expressed by the
LTL formula $\varphi_f^\R = \bo (\nrr \im \diam \nro)$.
We say that a protocol is fair, if in all instances of the protocol,
fairness for both $\O$ and $\R$ holds.
Hence the fairness requirement for the protocol is expressed
by the formula
\begin{align}
\varphi_f = \varphi_f^\O \wedge \varphi_f^\R \eqpun . 
\label{objective-f}
\end{align}

\medskip\noindent{\bf Abuse-freeness.\/}
The definition of abuse-freeness as given in \cite{GarayJM99}, 
is the following:
{\em ``An optimistic contract signing protocol is abuse-free if it is 
impossible for a single player at any point in the protocol to be able 
to prove to an outside party that he has the power to terminate 
(abort) or successfully complete the contract''\/}.
In \cite{ChadhaMSS03}, the authors prove that in any fair optimistic 
protocol, an optimistic participant yields an advantage to the other 
participant.
In a given protocol instance, once an agent has the other agent's intent 
to sign a contract, he can use this intent to negotiate a different 
contract with a third party, while ensuring that the original protocol
instance is aborted.
The term aborted is used here to mean that neither agent can get 
a non-repudiation evidence in a given protocol instance, once that 
instance is aborted.
As noted by the authors of \cite{ChadhaMSS03}, the best that 
one can hope for is to prevent either participant from proving to a 
third party that he has an advantage, or in other words, that he has
the other participant's intent to sign the contract.
This is defined as {\em abuse-freeness\/}.
As noted by the authors of \cite{GarayJM99,Kremer02gameanalysis}, using 
PCS or {\em Private Contract Signatures\/}, introduced by Garay et al., 
in \cite{GarayJM99}, which provides the designated verifier property, 
neither agent can prove the other agent's intent to sign the contract 
to anyone other than the $\ttp$.
Therefore, ensuring abuse-freeness requires the use of PCS.
Since PCS are requisite to ensure abuse-freeness, we do not model
abuse-freeness, or the stronger property balance \cite{ChadhaKS01},
in our formalism.

\medskip\noindent{\bf Timeliness.\/}
Informally, timeliness is defined as follows: {\em ``A protocol 
respects timeliness, if both agents always have the ability to reach, 
in a finite amount of time, a point in the protocol where they can stop
the protocol while preserving fairness''\/}.
We do not model timeliness in this paper as the cases in the 
literature where timeliness is compromised involve the lack of
an abort subprotocol.
Since we explicitly include the capability to abort the protocol,
our solution provides timeliness as guaranteed by existing
protocols.
Alternatively, timeliness could be explicitly modeled in the 
specifications of the agents and the $\ttp$, but in the interest 
of keeping the objectives simpler so that we convey the more 
interesting idea of using assume-guarantee synthesis, we avoid 
modeling timeliness explicitly.

\medskip\noindent{\bf Signature exchange.\/}
A protocol is an exchange protocol if it enables the exchange
of signatures.
This is also referred to as {\em Viability\/} in the literature.
For an exchange protocol to be a non-repudiation protocol,
at the end of every run of the protocol, either the agents 
have their respective non-repudiation evidences, or, if they
do not have their non-repudiation evidences, they have the abort 
token.
The property that evidences once obtained are not repudiable is
referred to as {\em Non-repudiability\/}.
A fair non-repudiation protocol must satisfy fairness,
abuse-freeness, non-repudiability and viability.

We now present intuitive objectives for the agents and the
trusted third party and show that satisfaction of these objectives 
implies that the protocols we synthesize are fair.

\medskip\noindent{\bf Specification for the originator $\O$.}
The objective of the originator $\O$ is expressed as follows:
\begin{itemize}
\item In all protocol instances, she eventually sends the evidence 
of origin.
This is expressed by the LTL formula $\varphi_\O^1 = \diam M_1$.
\item In all protocol instances, one of the following statements 
should be true:
\begin{enumerate}
\item (a) The originator eventually gets the recipient's signature 
$\sigb^\R$ or, (b) she eventually gets the recipient's signature 
$\sigb^\ttp$ and never gets the abort token $\abo$.
This is expressed by the LTL formula 
\begin{align*}
\varphi_\O^2 = (\diam \sigb^\R \vee 
(\diam \sigb^\ttp \wedge \bo \neg \abo))
\eqpun .
\end{align*}
\item (a) The originator eventually gets the abort token and, 
(b) the recipient never gets her signature $\siga^\O$ and never gets 
her signature $\siga^\ttp$ from the $\ttp$.
This is expressed by the LTL formula
\begin{align*}
\varphi_\O^3 = 
\diam \abo \wedge (\bo \neg \siga^\O \wedge \bo \neg \siga^\ttp) =
\diam \abo \wedge \bo (\neg \siga^\O \wedge \neg \siga^\ttp)
\eqpun .
\end{align*}
\end{enumerate}
\end{itemize}
The objective $\varphi_\O$ of $\O$ can therefore be expressed by the 
following LTL formula
\begin{align}
\varphi_\O = \varphi_\O^1 \wedge \bo (\varphi_\O^2 \vee \varphi_\O^3)
\eqpun . \label{objective-o}
\end{align}
There are two interpretations of the abort token in the 
literature.
On the one hand the abort token was never intended to serve as
a proof that a protocol instance was not successfully completed;
it was to guarantee that the $\ttp$ would never resolve a protocol
after it has been aborted.
On the other hand, there is mention of the abort token being used by
the recipient to prove that the protocol was aborted.
We take the position that the abort token may be used to ensure
$\ttp$ accountability as noted in \cite{ShmatikovM02} and hence 
include it in the objective of $\O$. 
If the $\ttp$ misbehaves and issues both $\sigb^\ttp$ and
$\abo$, we claim that the objective $\varphi_\O$ of the originator
should be violated, but in this case, she has the power to 
prove that the $\ttp$ misbehaved by presenting both $\sigb^\ttp$ 
and $\abo$ to demonstrate inconsistent behavior.
While having both $\sigb^\ttp$ and $\abo$ is a violation of
$\varphi_\O$, having both $\sigb^\R$ and $\abo$ is not a violation
of $\varphi_\O$; once $\O$ receives $\sigb^\R$, we take it that the 
objective $\varphi_\O$ is satisfied.
While having both $\sigb^\R$ and $\sigb^\ttp$ may be interpreted
as $\O$ having inconsistent signatures, we do not consider this
to be a violation of $\O$'s objective; given the nature of
asynchronous networks it may well be the case that both these 
evidences arrive eventually, one from the $\ttp$ and the other from $\R$,
as $\O$ did not wait long enough before sending $r_1^\O$.

\medskip\noindent{\bf Specification for the recipient $\R$.}
The objective of the recipient $\R$ can be expressed as follows:
\begin{itemize}
\item In all protocol instances, if he gets the evidence of origin 
$\eoo$, then one of the following statements should be true:
\begin{enumerate}
\item (a) The recipient eventually gets the originator's signature 
$\siga^\O$ or, (b) he eventually gets the originator's signature
$\siga^\ttp$ and never gets the abort token $\abr$.
This is expressed by the LTL formula
\begin{align*}
\varphi_\R^1 = (\diam \siga^\O \vee 
(\diam \siga^\ttp \wedge \bo \neg \abr)) \eqpun .
\end{align*}
\item (a) The recipient eventually gets the abort token and, 
(b) the originator never gets his signature $\sigb^\R$ and never gets 
his signature $\sigb^\ttp$ from the $\ttp$.
This is expressed by the LTL formula
\begin{align*}
\varphi_\R^2 = 
\diam \abr \wedge (\bo \neg \sigb^\R \wedge \bo \neg \sigb^\ttp) =
\diam \abr \wedge \bo (\neg \sigb^\R \wedge \neg \sigb^\ttp) 
\eqpun .
\end{align*}
\end{enumerate}
\end{itemize} 
The objective $\varphi_\R$ can therefore be expressed by the LTL formula
\begin{align}
\varphi_\R = \bo (\eoo \im (\varphi_\R^1 \vee \varphi_\R^2))
\eqpun . \label{objective-r}
\end{align}
If the $\ttp$ misbehaves and issues both $\siga^\ttp$ and
$\abr$, we claim that the objective $\varphi_\R$ of the recipient
should be violated, but in this case he has the power to 
prove that the $\ttp$ misbehaved by presenting both $\siga^\ttp$ 
and $\abr$.
Similar to the case of $\O$, once $\R$ receives $\siga^\O$, the 
objective $\varphi_\R$ is satisfied whether or not abort tokens or 
non-repudiation evidences are issued by the $\ttp$.

\medskip\noindent{\bf Specification for the trusted third party $\ttp$.}
The objective of the trusted third party is to treat both agents 
symmetrically and be accountable to both agents.
It can be expressed as follows:
\begin{itemize}
\item In all protocol instances, if the abort request $a_1^\O$ or
a resolve request $r_1^\O$ or $r_1^\R$ is received, then eventually the 
$\ttp$ sends the abort token $\abo$ or the abort token $\abr$ or the 
originator's signature $\siga^\ttp$ or the recipient's signature 
$\sigb^\ttp$.
This can be expressed by the LTL formula
\begin{align*}
\varphi_\ttp^1 =
\bo ((\reqa \vee \reqr) \im (\diam \abo \vee \diam \abr \vee
\diam \siga^\ttp \vee \diam \sigb^\ttp))
\eqpun .
\end{align*}
\item In all protocol instances, if the originator's signature
$\siga^\ttp$ has been sent to the recipient, then the originator
should eventually get the recipient's signature $\sigb^\ttp$ 
and the agents should never get the abort token.
This can be expressed by the LTL formula
\begin{align*}
\varphi_\ttp^2 = 
\bo (\siga^\ttp \im 
(\diam \sigb^\ttp \wedge 
\bo (\neg \abo \wedge \neg \abr))) \eqpun .
\end{align*}
\item Symmetrically, in all protocol instances, if the recipient's 
signature $\sigb^\ttp$ has been sent to the originator, then the 
recipient should eventually get the originator's signature $\siga^\ttp$ 
and the agents should never get the abort token.
This can be expressed by the LTL formula
\begin{align*}
\varphi_\ttp^3 = 
\bo (\sigb^\ttp \im 
(\diam \siga^\ttp \wedge 
\bo (\neg \abo \wedge \neg \abr))) \eqpun .
\end{align*}
\item In all protocol instances, if the originator gets the
abort token $\abo$, then the recipient should eventually
get the abort token $\abr$ and the originator should never get
the recipient's signature $\sigb^\ttp$ and the recipient should
never get the originator's signature $\siga^\ttp$.
This can be expressed by the LTL formula
\begin{align*}
\varphi_\ttp^4 = 
\bo (\abo \im 
(\diam \abr \wedge 
\bo (\neg \siga^\ttp \wedge \neg \sigb^\ttp))) \eqpun .
\end{align*}
\item Symmetrically, in all protocol instances, if the recipient 
gets the abort token $\abr$, then the originator should eventually
get the abort token $\abo$ and the originator should never get
the recipient's signature $\sigb^\ttp$ and the recipient should
never get the originator's signature $\siga^\ttp$.
This can be expressed by the LTL formula
\begin{align*}
\varphi_\ttp^5 = 
\bo (\abr \im 
(\diam \abo \wedge 
\bo (\neg \siga^\ttp \wedge \neg \sigb^\ttp))) \eqpun .
\end{align*}
\end{itemize}
The objective $\varphi_\ttp$ of the $\ttp$ is then defined as:
\begin{align}
\varphi_\ttp = \varphi_\ttp^1 \wedge \varphi_\ttp^2 \wedge
\varphi_\ttp^3 \wedge \varphi_\ttp^4 \wedge \varphi_\ttp^5
\eqpun . \label{objective-ttp}
\end{align}
Note that our objective for the $\ttp$ treats both 
agents symmetrically.
In this paper we present assume-guarantee synthesis for the above
objective of the $\ttp$.
But in general, the objective of the $\ttp$ can be weakened if
desired, by treating the agents asymmetrically, and the assume-guarantee
synthesis technique can be applied with this weakened objective.
We remark that the specifications of the participants in our protocol 
model are sequences of messages.
Using predicates that are set when messages are sent or received by the
agents or the $\ttp$, we transform those informal specifications 
into formal objectives using the predicates and LTL.
The following theorem shows that the objectives we have introduced
(\ref{objective-o}), (\ref{objective-r}) and
(\ref{objective-ttp}) imply fairness (\ref{objective-f}).

\begin{theo}{\ (Objectives imply fairness)} 
\label{objectives-imply-ft}
We have,
$\varphi_\O \wedge \varphi_\R \wedge \varphi_\ttp \im \varphi_f$.
\end{theo}

\begin{proof}
To prove the assertion, assume towards a contradiction that there exists
a path that satisfies $\varphi_\O \wedge \varphi_\R \wedge \varphi_\ttp$ 
but does not satisfy $\varphi_f$.
We consider the case when the path does not satisfy the first conjunct
$\varphi_f^\O = \bo (\nro \im \diam \nrr)$ (a similar
argument applies to the second conjunct).
If the path does not satisfy $\varphi_f^\O$, then there is a suffix of 
the path, where $\eoo \wedge (\siga^\O \vee \siga^\ttp)$ holds but 
$\eor \wedge (\sigb^\R \vee \sigb^\ttp)$ does not hold at all states 
of the suffix.
It follows that the path satisfies
\begin{align}
\diam \bo (\eoo \wedge (\siga^\O \vee \siga^\ttp) \wedge 
(\neg \eor \vee (\neg \sigb^\R \wedge \neg \sigb^\ttp))) \eqpun .
\label{theo-1-a-1}
\end{align}
Consider the objective
$\varphi_\O^2 = (\diam \sigb^\R \vee (\diam \sigb^\ttp \wedge
\bo \neg \abo))$.
Since all predicates are monotonic, we can rewrite $\varphi_\O^2$
as follows:
\[
\varphi_\O^2 = \diam \bo (\sigb^\R \vee (\sigb^\ttp \wedge
\neg \abo)) \eqpun .
\]
Similarly, we can rewrite $\varphi_\O^3$ as follows:
\[
\varphi_\O^3 = \diam \bo (\abo \wedge \neg \siga^\O \wedge
\neg \siga^\ttp) \eqpun .
\]
If a path satisfies (\ref{theo-1-a-1}), then it also satisfies
$\diam \bo (\siga^\O \vee \siga^\ttp)$.
By the monotonicity of the predicates, we have 
$\diam \bo (\siga^\O \vee \siga^\ttp)$ is equivalent to 
$\diam \bo \siga^\O \vee \diam \bo \siga^\ttp$.
We consider the following cases to complete the proof:
\begin{enumerate}
\item {\em Case 1. Path satisfies $\diam \bo \siga^\O$.\/}
If the path satisfies $\diam \bo \siga^\O$, then the path does not
satisfy $\varphi_\O^3$.
We now show that the path also does not satisfy $\varphi_\O^2$.
Since the path satisfies $\diam \bo \siga^\O$, it must be the case that 
message $m_2$ was received by $\O$, as otherwise $\O$ will not send 
$\siga^\O$.
This implies that the path satisfies $\diam \bo \eor$.
Since the path satisfies both $\diam \bo \eor$ and (\ref{theo-1-a-1}), 
it follows that the path must satisfy 
$\diam \bo (\neg \sigb^\R \wedge \neg \sigb^\ttp)$.
Hence the path does not satisfy $\diam \bo \sigb^\R$ and 
$\diam \bo \sigb^\ttp$ leading to the path violating $\varphi_\O^2$.
Since the path does not satisfy both $\varphi_\O^2$ and $\varphi_\O^3$, 
it does not satisfy $\varphi_\O$, which is a contradiction.
\item {\em Case 2. Path satisfies $\diam \bo \siga^\ttp$.\/}
If the path satisfies $\diam \bo \siga^\ttp$, then either $\O$ or 
$\R$ must have sent the resolve request.
If the $\ttp$ resolves the protocol only to the agent that sends
the resolve request and not the other, then the path does not 
satisfy $\varphi_\ttp$, leading to a contradiction.
For $\varphi_\ttp$ to hold, the $\ttp$ must have sent both
$\siga^\ttp$ and $\sigb^\ttp$, which given the channels between the
agents and the $\ttp$ are resilient implies, (a) $\eor$ must have
been set by $\O$ upon receiving $\sigb^\ttp$ leading to the path
satisfying $\diam \bo \eor$ and (b) the path satisfies 
$\diam \bo \sigb^\ttp$.
Since the path satisfies $\diam \bo \eor$ and $\diam \bo \sigb^\ttp$,
it cannot satisfy (\ref{theo-1-a-1}), leading to a contradiction.
\end{enumerate}
\qed
\end{proof}

%% file: definitions.tex
\section{Co-synthesis}
\label{sec:co-synthesis}

In this section we first define processes, schedulers and
objectives for synthesis along the lines of \cite{CH07}.
Next we define traditional co-operative~\cite{CE81} and
strictly competitive~\cite{PnueliRosner89,RamadgeWonham87} versions
of the co-synthesis problem; we refer to them as 
\emph{weak co-synthesis} and \emph{classical co-synthesis}, respectively.
We then define a formulation of co-synthesis introduced in
\cite{CH07} called \emph{assume-guarantee synthesis}.
We show later in the paper that the protocol model of 
Section~\ref{sec:fair-nonrep} reduces to the process model for
synthesis that we present in this section.

\medskip\noindent{\bf Variables, valuations, and traces.}
Let $X$ be a finite set of variables such that each variable
$x \in X$ has a finite domain $D_x$.
A \emph{valuation} $\valua$ on $X$ is a function
$\valua: X \to \bigcup_{x \in X} D_x$ that assigns to each variable
$x \in X$ a value $\valua(x) \in D_x$.
We write $\Valua[X]$ for the set of valuations on $X$.
A \emph{trace} on $X$ is an infinite sequence
$(v_0,v_1,v_2,\ldots) \in \Valua[X]^\omega$ of valuations on $X$.
Given a valuation $\valua[X] \in \Valua[X]$ and a subset 
$Y \subseteq X$ of the variables, we denote by 
$\valua[X] \obciach Y$ the restriction of the valuation $\valua[X]$ to 
the variables in $Y$.
Similarly, for a trace $\trace(X) = (v_0,v_1,v_2,\ldots)$ on $X$, we write
$\trace(X) \obciach Y = (v_0 \obciach Y, v_1 \obciach Y, v_2 \obciach Y,
\ldots)$ for the restriction of $\trace(X)$ to the variables in $Y$.
The restriction operator is lifted to sets of valuations, and to
sets of traces.

\medskip\noindent{\bf Processes and refinement.}
Let $\moves$ be a finite set of moves.
For $i \in \set{1,2,3}$, a \emph{process} is defined by the
tuple $P_i=(X_i,\mov_i,\trans_i)$ where,
\begin{enumerate}
\item $X_i$ is a finite set of variables of process $P_i$ with
$X = \bigcup_{i = 1}^3 X_i$ being the set of all process variables,
\item $\mov_i: \Valua_i[X_i] \to 2^\moves \setminus \emptyset$ is a move
assignment that given a valuation in $\Valua_i[X_i]$, returns a 
non-empty set of moves, where $\Valua_i[X_i]$ is the set of 
valuations on $X_i$, and
\item $\trans_i: \Valua_i[X_i] \times \moves \to 2^{\Valua_i[X_i]} 
\setminus \emptyset$ is a non-deterministic transition function.
\end{enumerate}
The set of process variables $X$ may be shared between processes.
The processes only choose amongst available moves at every valuation of 
their variables as determined by their move assignment.
The transition function maps a present valuation and a process move to 
a nonempty set of possible successor valuations such that each successor 
valuation has a unique pre-image.
The uniqueness of the pre-image is a property of fair exchange
protocols; unique messages convey unique content and generate unique
valuations.

A \emph{refinement} of process $P_i=(X_i,\mov_i,\trans_i)$ is a process
$P_i'=(X_i',\mov_i',\trans_i')$ such that:
\begin{enumerate}
\item $X_i \subseteq X_i'$,
\item for all valuations $\valua_i[X_i']$ on $X_i'$, we have
$\mov_i'(\valua_i[X_i']) \subseteq \mov_i(\valua_i[X_i'] \obciach X_i)$, 
and
\item for all valuations $\valua_i[X_i']$ on $X_i'$ and for all moves
$a \in \mov_i'(\valua_i[X_i'])$, we have
$\trans_i'(\valua_i[X_i'], a) \obciach X_i \subseteq 
\trans_i(\valua_i[X_i'] \obciach X_i, a)$.
\end{enumerate}
In other words, the refined process $P_i'$ has possibly more variables
than the original process $P_i$, at most the same moves as the moves of
the original process $P_i$ at every valuation, and every possible update
of the variables in $X_i$ given $\mov_i'$ by $P_i'$ is a possible update 
by $P_i$.
We write $P_i' \preceq P_i$ to denote that $P_i'$ is a refinement of $P_i$.
Given refinements $P_1'$ of $P_1$, $P_2'$ of $P_2$ and $P_3'$ of $P_3$,
we write $X' = X_1' \cup X_2' \cup X_3'$ for the set of variables of all
refinements, and we denote the set of valuations on $X'$ by $\Valua[X']$.

\medskip\noindent{\bf Schedulers.}
Given processes $P_i$, where $i \in \set{1, 2, 3}$, a \emph{scheduler} 
$\Sc$ for $P_i$ chooses at each computation step whether it is 
process $P_1$'s turn, process $P_2$'s turn or process $P_3$'s turn to 
update her variables.
Formally, the scheduler $\Sc$ is a function 
$\Sc: \Valua[X]^* \to \set{1,2,3}$
that maps every finite sequence of global valuations
(representing the history of a computation)
to $i\in\set{1,2,3}$, signaling that process $P_i$ is next to update her
variables.
The scheduler $\Sc$ is \emph{fair} if it assigns turns to $P_1$, $P_2$ and
$P_3$ infinitely often;
i.e., for all traces $(v_0,v_1,v_2,\ldots) \in \Valua[X]^\omega$, there
exist infinitely many $j_i \geq 0$, such that
$\Sc(v_0, \ldots, v_{j_1}) = 1$, $\Sc(v_0, \ldots, v_{j_2}) = 2$ and
$\Sc(v_0, \ldots, v_{j_3}) = 3$.
Given three processes $P_1 = (X_1, \mov_1, \trans_1)$, 
$P_2 = (X_2, \mov_2, \trans_2)$ and $P_3 = (X_3, \mov_3, \trans_3)$,
a scheduler $\Sc$ for $P_1$, $P_2$ and $P_3$, and a start valuation
$v_0 \in \Valua[X]$, the set of possible traces is:
\begin{align*}
\semb{(P_1 \para P_2 \para P_3 \para \Sc)(v_0)}  = 
\set{ & (v_0, v_1, v_2, \ldots) \in \Valua[X]^\omega \mid
\forall j \geq 0. \ \Sc(v_0, \ldots, v_j) = i; \\
& v_{j+1} \obciach (X \setminus X_i) = v_j \obciach (X \setminus X_i); \\
& v_{j+1} \obciach X_i \in \trans_i(v_j \obciach X_i, a)
\text{ for some }
a \in \mov_i(v_j \obciach X_i))} \eqpun .
\end{align*}
Note that during turns of one process $P_i$, the values of the private
variables $X \setminus X_i$ of the other processes remain unchanged.
We define the projection of traces to moves as follows:
\begin{align*}
(v_0, v_1, v_2, \ldots) \obciach \moves =
\set{ & (a_0, a_1, a_2, \ldots) \in \moves^\omega \mid 
\forall j \geq 0. \ \Sc(v_0, \ldots, v_j) = i; \\
& v_{j+1} \obciach X_i \in \trans_i(v_j \obciach X_i, a_j);
a_j \in \mov_i(v_j \obciach X_i)}
\eqpun .
\end{align*}

\medskip\noindent{\bf Specifications.}
A \emph{specification} $\varphi_i$ for process $P_i$ is a set of traces
on $X$;
that is, $\varphi_i \subseteq \Valua[X]^\omega$.
We consider only $\omega$-regular specifications~\cite{Thomas97}.
We define boolean operations on specifications using logical operators
such as $\wedge$ (conjunction) and $\im$ (implication).

The input to the co-synthesis problem is given as follows:
for $i \in \set{1, 2, 3}$, processes $P_i = (X_i, \mov_i, \trans_i)$,
specifications $\varphi_i$ for process~i, and a start valuation 
$v_0 \in \Valua$.

\medskip\noindent{\bf Weak co-synthesis.}
The \emph{weak co-synthesis} problem is defined as follows:
do there exist refinements $P_i'=(X_i',\mov_i', \trans_i')$ and a 
valuation $v_0' \in \Valua'$,
such that,
\begin{enumerate}
\item $P_i' \preceq P_i$ and $v_0' \obciach X =v_0$, and 
\item For all fair schedulers $\Sc$ for $P_i'$ we have,\\
$\semb{(P_1' \para P_2' \para P_3' \para \Sc)(v_0')} \obciach X
\subseteq (\varphi_1 \wedge \varphi_2 \wedge \varphi_3).$
\end{enumerate}
Intuitively, weak co-synthesis or co-operative co-synthesis is
a synthesis formulation that seeks refinements $P_1'$, $P_2'$ and $P_3'$ 
where the processes co-operate to satisfy their respective objectives.

\medskip\noindent{\bf Classical co-synthesis.}
The \emph{classical co-synthesis} problem is defined as follows:
do there exist refinements $P_i'=(X_i', \mov_i', \trans_i')$ and
a valuation $v_0' \in \Valua'$, such that,
\begin{enumerate}
\item $P_i' \preceq P_i$ and $v_0' \obciach X =v_0$, and
\item For all fair schedulers $\Sc$ for $P_i'$ we have,
\begin{enumerate}
\item $\semb{(P_1' \para P_2 \para P_3 \para \Sc)(v_0')} \obciach X 
\subseteq \varphi_1$;
\item $\semb{(P_1 \para P_2' \para P_3 \para \Sc)(v_0')} \obciach X 
\subseteq \varphi_2$;
\item $\semb{(P_1 \para P_2 \para P_3' \para \Sc)(v_0')} \obciach X 
\subseteq \varphi_3$.
\end{enumerate}
\end{enumerate}
Classical or strictly competitive co-synthesis is a formulation
that seeks refinements $P_1'$, $P_2'$ and $P_3'$ such that $P_1'$
can satisfy $\varphi_1$ against all possible, and hence adversarial,
behaviors of the other processes; similarly for $P_2'$ and $P_3'$.

\medskip\noindent{\bf Assume-guarantee synthesis.}
The \emph{assume-guarantee synthesis} problem is defined as
follows:
do there exist refinements $P_i'=(X_i', \mov_i', \trans_i')$ 
and a valuation $v_0' \in \Valua'$, such that,
\begin{enumerate}
\item $P_i' \preceq P_i$ and $v_0' \obciach  X = v_0$,
and
\item For all fair schedulers $\Sc$ for $P_i'$ we have,
\begin{enumerate}
\item $\semb{(P_1' \para P_2 \para P_3 \para \Sc)(v_0')} \obciach X
\subseteq (\varphi_2 \wedge \varphi_3) \im \varphi_1$;
\item $\semb{(P_1 \para P_2' \para P_3 \para \Sc)(v_0')} \obciach X
\subseteq (\varphi_1 \wedge \varphi_3) \im \varphi_2$;
\item $\semb{(P_1 \para P_2 \para P_3' \para \Sc)(v_0')} \obciach X
\subseteq (\varphi_1 \wedge \varphi_2) \im \varphi_3$;
\item $\semb{(P_1' \para P_2' \para P_3' \para \Sc)(v_0')} \obciach X
\subseteq (\varphi_1 \wedge \varphi_2 \wedge \varphi_3)$.
\end{enumerate}
\end{enumerate}
Assume-guarantee synthesis or conditionally competitive co-synthesis
is a formulation that seeks refinements $P_1'$, $P_2'$ and $P_3'$
such that $P_1'$ can satisfy $\varphi_1$ as long as processes $P_2$
and $P_3$ satisfy their objectives; similarly for $P_2'$ and $P_3'$.
This synthesis formulation is well suited for those cases 
where processes are primarily concerned with satisfying their own 
objectives and only secondarily concerned with violating the objectives 
of the other processes.
We want protocols to be correct under {\em arbitrary\/} behaviors of 
the participants, and the arbitrary or worst case behavior of a 
participant without sabotaging her own objective, is to first satisfy her 
own objective, and only then to falsify the objectives of the other 
participants. 
The primary goal of satisfying her own objective, and secondary goal of 
falsifying other participant objectives formally captures this 
worst case or arbitrary behavior assumption.
We show that this synthesis formulation is the only one that works for
fair non-repudiation protocols.
While classical co-synthesis can be solved as zero-sum games,
assume-guarantee synthesis can be solved using non zero-sum games
with lexicographic objectives \cite{CH07}. 
For brevity, we drop the initial valuation $v_0$ in the set of traces.

%% file: cosynthesis.tex
\section{Protocol Co-synthesis}
\label{sec:protocol-cosynth}

In this section, we present our results on synthesizing fair
non-repudiation protocols.
We use the process model in Section~\ref{sec:co-synthesis} to 
define agent and $\ttp$ processes, with objectives as defined in 
Section~\ref{sec:specifications}.
We then introduce the protocol synthesis model and show that
classical co-synthesis fails and weak co-synthesis generates 
unacceptable solutions.
We provide a game theoretic justification of the need for a
$\ttp$ by showing that without the $\ttp$ neither classical
co-synthesis nor assume-guarantee synthesis can be used to
synthesize fair non-repudiation protocols.
We define the set $P_{AGS}$ of assume-guarantee refinements and
prove that the refinements are attack-free.
We then present an alternate characterization of the set $P_{AGS}$
and show that the Kremer-Markowitch (KM) non-repudiation 
protocol with offline $\ttp$, proposed in 
\cite{MarkowitchK01,KremerMZ02,KremerR03}, is included in $P_{AGS}$ 
whereas the ASW certified mail protocol and the GJM protocol are not.
Finally, we systematically analyze refinements of the most general
agents and the $\ttp$ with respect to their membership in $P_{AGS}$
and show the KM protocol can be automatically generated.

\psfrag{i}{$\idle$}
\psfrag{m1!}{$m_1!$}
\psfrag{m2!}{$m_2!$}
\psfrag{m3!}{$m_3!$}
\psfrag{m4!}{$m_4!$}
\psfrag{m1?}{$m_1?$}
\psfrag{m2?}{$m_2?$}
\psfrag{m3?}{$m_3?$}
\psfrag{m4?}{$m_4?$}
\psfrag{a1!o}{$a_1^\O!$}
\psfrag{r1!o}{$r_1^\O!$}
\psfrag{a1!r}{$a_1^\R!$}
\psfrag{r1!r}{$r_1^\R!$}
\psfrag{a1!}{$a_1^*!$}
\psfrag{r1!}{$r_1^*!$}
\psfrag{a1!}{$a_1^*!$}
\psfrag{r1!}{$r_1^*!$}
\psfrag{a2?}{$a_2^*?$}
\psfrag{r2?}{$r_2^*?$}

\begin{figure}[!t]
\centering
\includegraphics[width=4.5in,height=3in]{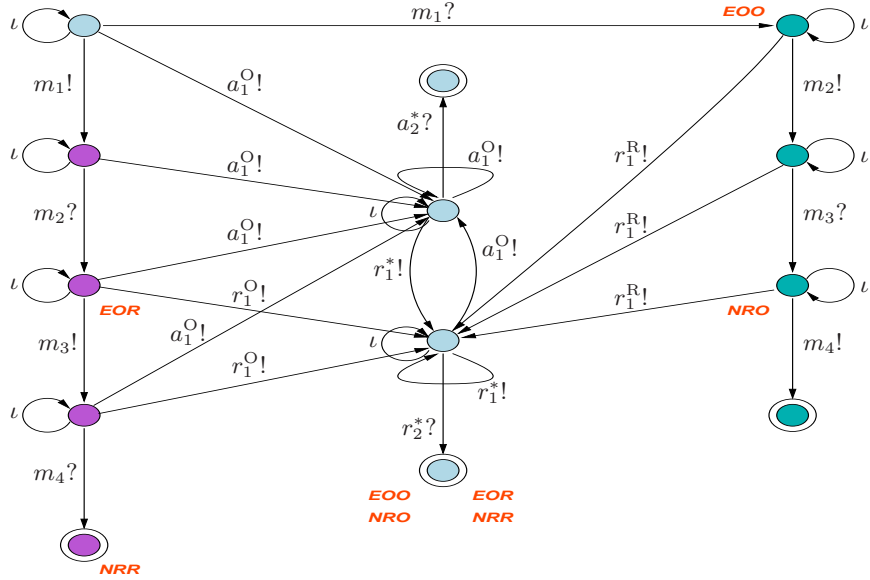}
\caption{An interface automaton that shows the states and 
enabled moves of the agents $\O$ (on the left) and $\R$ 
(on the right).
Move $\idle$ is the idle move. 
The states with no outgoing edges are terminal. 
We consider the most liberal behaviors of the agents wherein the 
abort and resolve messages can be sent from all states where the 
agents have the data they need to send those messages. 
The predicates are monotonic and are shown in the first state at 
which they hold.
In states that can be either agent state, we use the $^*$ in
the messages $a_2^*, r_1^*, r_2^*$ to denote one of $\O$
or $\R$.
Abort or resolve requests can be sent from the states
marked terminal, but they have no bearing on the outcome of the 
protocol and hence we omit them.}
\label{figure:agents}
\end{figure}

\medskip\noindent{\bf The process $\O$.}
We distinguish between the set of messages sent by $\O$ and the set
of messages received by $\O$.
We first recall, from Section~\ref{sec:fair-nonrep}, that $\O$ sets 
the predicates $\eor$, $\sigb^\R$, $\sigb^\ttp$ and $\abo$ when she 
receives messages $m_2$, $m_4$, $r_2^\O$ and $a_2^\O$ respectively. 
We add to this set the predicates $M_1$, $M_3$, $\reqa^\O$ and
$\reqr^\O$ that are set by $\O$ when she sends messages 
$m_1$, $m_3$, $a_1^\O$ and $r_1^\O$ respectively.
We take the set of variables of the process $\O$ as
$X_\O = \set{M_1, \eor, M_3, \sigb^\R, \sigb^\ttp, \reqa^\O,
\reqr^\O, \abo}$; the union of the predicates set by $\O$ when she 
receives messages and the set of predicates set by $\O$ when she 
sends messages.
By an abuse of notation, we take the set of all messages that can
be sent by $\O$ as the moves of process $\O$.
By including an idle move $\idle$, which $\O$ may choose in lieu of
sending a message, we get the following set of moves for $\O$:
$\moves_\O = \set{\idle, m_1, m_3, a_1^\O, r_1^\O}$.

\medskip\noindent{\bf The process $\R$.}
Similar to the case of process $\O$, we define the set of 
variables of process $\R$ as the union of the set of predicates set by 
$\R$ when he sends messages and the set of predicates he sets when
he receives messages.
We have the predicates $\eoo$, $\siga^\O$, $\siga^\ttp$ and $\abr$,
set by $\R$ when he receives messages $m_1$, $m_3$, $r_2^\R$
and $a_2^\R$ respectively.
We add to this the predicates $M_2$, $M_4$ and $\reqr^\R$, set by $\R$
when he sends messages $m_2$, $m_4$ and $r_1^\R$ respectively.
This gives us the following variables for process $\R$:
$X_\R = \set{\eoo, M_2, \siga^\O, M_4, \siga^\ttp, \reqr^\R, 
\abr}$.
The set of moves for $\R$ is given by 
$\moves_\R = \set{\idle, m_2, m_4, r_1^\R}$.
In Figure~\ref{figure:agents}, we show an interface 
automaton for an agent.
Since an agent can act either as an originator or a recipient,
we show the actions available to the agent in both roles in the figure.

\medskip\noindent{\bf The process $\ttp$.}
The predicates $\reqa$ and $\reqr$ are set by the $\ttp$ when she
receives an abort or a resolve request from either agent.
We add to this the predicates $A_2^\O$, $A_2^\R$, $R_2^\O$ and
$R_2^\R$, set by the $\ttp$ when she sends messages $a_2^\O$,
$a_2^\R$, $r_2^\O$ and $r_2^\R$ respectively.
We get the following set of process variables for the $\ttp$:
$X_\ttp = \set{\reqa, \reqr, A_2^\O, A_2^\R, R_2^\O, R_2^\R}$.
The set of moves for the $\ttp$ are defined as follows:
$\moves_\ttp = \set{\idle, a_2^\O, a_2^\R, \comp{a_2^\O, a_2^\R},
r_2^\O, r_2^\R, \comp{r_2^\O, r_2^\R}}$.
The $\ttp$ move $\comp{a_2^\O, a_2^\R}$
results in the $\ttp$ sending messages $a_2^\O$ to $\O$ and 
$a_2^\R$ to $\R$.
The $\ttp$ can choose to send them in any order; all that is 
guaranteed is that both messages will be sent by the $\ttp$.
Similarly for the $\ttp$ move $\comp{r_2^\O, r_2^\R}$.
The moves for the $\ttp$ are shown in 
Table~\ref{table-ttp-moves};
these include the moves for the $\ttp$ in the ASW certified
mail protocol \cite{AsokanSW98}, the GJM protocol \cite{GarayJM99}
and the KM protocol \cite{MarkowitchK01}.
We show moves for the $\ttp$ with and without a persistent database
for completeness.
Since it is trivially the case that $\ttp$ accountability cannot be 
satisfied without a persistent database, we do not consider the absence 
of a persistent database in the rest of this paper.
\begin{table*}[!t]
\begin{center}
\begin{tabular}{|l||l|l|l|l|l|}
\hline
\bf{Agent moves} & \multicolumn{5}{c|}{\bf{Enabled $\ttp$ moves}} \\
\cline{2-6}
& \multicolumn{2}{c|}{\bf{Without DB}} & \multicolumn{3}{c|}{\bf{With a persistent DB}} \\
\cline{4-6}
& \multicolumn{2}{c|}{} & \bf{ASW} & \bf{GJM} & \bf{KM} \\
\hline \hline
\bf{$\O$ sends $a_1^\O$} & $a_2^\O$ & $\comp{a_2^\O, a_2^\R}$ & If $\R$ has recovered, & If recovered, then $r_2^\O$ & If aborted or \\
& & & invite $\O$ to recover & else $a_2^\O$ & recovered, then \\
& & & else $a_2^\O$ & & $\idle$ else $\comp{a_2^\O, a_2^\R}$ \\
\hline
\bf{$\O$ sends $r_1^\O$} & $r_2^\O$ & $\comp{r_2^\O, r_2^\R}$ & If aborted, then $a_2^\O$ & If aborted, then $a_2^\O$ & If aborted or \\
& & & else $r_2^\O$ & else $r_2^\O$ & recovered, then \\
& & & & & $\idle$ else $\comp{r_2^\O, r_2^\R}$ \\
\hline
\bf{$\R$ sends $r_1^\R$} & $r_2^\R$ & $\comp{r_2^\O, r_2^\R}$ & If aborted, then $a_2^\R$ & If aborted, then $a_2^\R$ & If aborted or \\
& & & else $r_2^\R$ & else $r_2^\R$ & recovered, then \\
& & & & & $\idle$ else $\comp{r_2^\O, r_2^\R}$ \\
\hline
\end{tabular}
\end{center}
\caption[The moves available to the trusted third party]{In this
table we list the choices of moves available to the trusted third
party. Each row begins with a message sent by an agent to the $\ttp$
followed by the choices available to the $\ttp$ in all subsequent
states. The $\ttp$ moves for the ASW, GJM and KM protocols are shown.}
\label{table-ttp-moves}
\end{table*}

\medskip\noindent{\bf The protocol synthesis model.}
We now have all the ingredients to define our protocol synthesis 
model.
Given process $\O$, process $\R$ and process $\ttp$ as defined
above, we take $X = X_\O \cup X_\R \cup X_\ttp$ as the joint set of
process variables.
We take the objectives $\varphi_\O$, $\varphi_\R$ and 
$\varphi_\ttp$ for the processes $\O$, $\R$ and $\ttp$ respectively,
as defined in Section~\ref{sec:specifications}.
The set of traces $\semb{\O \para \R \para \ttp \para \Sc}$,
given $\Sc$ is a fair scheduler, is then the joint behavior of the 
most general agents and the most general $\ttp$, subject to the 
constraint that they can only send messages based on what they know 
at every valuation of their variables.
A protocol is a refinement $\O' \preceq \O$, $\R' \preceq \R$
and $\ttp' \preceq \ttp$, where each participant has a restricted
set of moves at every valuation of the process variables; the
restrictions constituting the rules of the protocol.
We take a protocol state as a valuation over the process 
variables.
By an abuse of notation, we represent every state of the protocol by the 
set of variables that are set to true in that state; for example a 
valuation $f = \set{M_1, \eoo, M_2, \eor}$ corresponds to the state of
the protocol after messages $m_1$ and $m_2$ have been received.
$f \obciach X_\R = \set{\eoo, M_2}$ 
corresponds to the restriction of the valuation $f$ to the variables 
of process $\R$; all that $\R$ knows in this state is that he has 
received $m_1$ and has sent $m_2$.
We take $v_0$ as the initial valuation where all variables are
false.
The set of variables in the refinements 
$\O' \preceq \O$, $\R' \preceq \R$ and $\ttp' \preceq \ttp$ are the
same as those in processes $\O$, $\R$ and $\ttp$,
respectively, and all traces begin with the initial valuation $v_0$.
We do not model the set of channels explicitly but reason against all
possible behaviors of unreliable channels.
We assume that every message at least includes the name
of the sender, is signed with the private key of the sender and 
encrypted with the public key of the recipient.

The following theorem states that the protocol model from 
Section~\ref{sec:fair-nonrep} and the protocol synthesis model presented
above are equivalent.
Let $A_0 = \O$, $A_1 = \R$ and $A_2 = \ttp$ be the most general 
participants with $A = \set{A_0, A_1, A_2}$ and variables 
$V_0 = X_\O$ for $A_0$, $V_1 = X_\R$ for $A_1$ 
and $V_2 = X_\ttp$ for $A_2$ as defined above.
It is then easy to show that,

\begin{theo}{\ (Trace equivalence of models)}
\label{theo-trace-equiv}
For all participant restrictions $A_i'$ and refinements 
$\O' \preceq \O$, $\R' \preceq \R$ and $\ttp' \preceq \ttp$, such that
$i \in \set{0, 1, 2}$ with $j = \O$ when $i = 0$, $j = \R$ when
$i = 1$ and $j = \ttp$ when $i = 2$, for all valuations
$v \in \calf[V_i]$, if $\mas_i'(v) = \mov_{j'}(v)$, then we have,
{\em Runs\/}$(\set{A_0',A_1',A_2'}) = 
\semb{\O' \para \R' \para \ttp' \para \Sc}$.
%
%
\end{theo}

We note that in Theorem~\ref{theo-trace-equiv}, when we say all 
restrictions $A_i'$ or all refinements $\O'$, $\R'$, and $\ttp'$, 
the most general participants are included 
(for example $\O'$ can be $\O$) and hence Theorem~\ref{theo-trace-equiv}
covers trace equivalence for all required cases.

\subsection{Failure of Classical and Weak Co-Synthesis}

In this subsection we show that classical co-synthesis fails 
while weak co-synthesis generates solutions that
are not attack-free and are hence unacceptable.
We first tackle classical co-synthesis.
In order to show failure of classical co-synthesis we need to
show that one of the following conditions:
\begin{enumerate}
\item $\semb{(\O' \para \R \para \ttp \para \Sc)} 
\subseteq \varphi_\O$;
\item $\semb{(\O \para \R' \para \ttp \para \Sc)}
\subseteq \varphi_\R$;
\item $\semb{(\O \para \R \para \ttp' \para \Sc)}
\subseteq \varphi_\ttp$,
\end{enumerate}
can be violated.
We show that for all refinements $\R'$ of the recipient $\R$, 
that is, for every sequence of moves ending in a move chosen by $\R'$, 
there exist moves for the other processes $\O$, $\ttp$ and $\Sc$,
and a behavior of the the channels, to extend that sequence such that 
the objective $\varphi_\R$ is violated.
Since $\R$ should satisfy his objective against all possible 
behaviors of the channels, to show failure of classical co-synthesis
it suffices to fix the behavior of all channels.
We assume the channels eventually deliver all messages.

\begin{theo}{\ (Classical co-synthesis fails for $\R$)} 
\label{classical-failure}
For all refinements $\R' \preceq \R$, the following assertion holds:
\[
\semb{\O \para \R' \para \ttp \para \Sc} \not \subseteq \varphi_\R.
\]
\end{theo}

\begin{proof}
We consider every valuation of the process 
variables and the set of all possible moves that can be selected by 
$\R$ at each valuation.
This defines all possible refinements of $\R$.
Since every valuation is the result of a finite sequence of
moves (messages) chosen (sent)  by the agents and the $\ttp$, it 
suffices to consider all possible finite sequences of messages 
received, ending in a message chosen by $\R$.
Let $\trace = (v_0, v_1, \ldots, v_n)$ be a finite sequence of 
valuations seen in a partial protocol run, where $v_0$ is the starting 
valuation.
Let $\sigma = \trace \obciach \moves = \seq{a_0, a_1, \ldots, a_{n-1}}$
be the corresponding sequence of $n$ moves seen in the run.
At the beginning of a protocol run, we have $\sigma = \emptyset$.
In the following, on a case by case basis, we show the sequence of 
moves seen in a partial protocol run, ending in a move chosen by 
$\R$, followed by moves for $\O$, $\ttp$ and $\Sc$ that leads 
to a violation of $\varphi_\R$.
\begin{enumerate}[]
\item R1: $\seq{m_1, \idle}$
  \begin{itemize}
  \item Whenever $Sc$ schedules $\O$, she chooses the idle action
    $\idle$.
    Since $\eoo$ is true, as long as $\O$ does not abort the protocol
    but chooses to remain idle, $\varphi_\R$ is violated.
  \item $\varphi_\R$ and $\varphi_\O$ are violated but $\varphi_\ttp$
    is satisfied.
  \end{itemize}
\item R2: $\seq{m_1, m_2}$
  \begin{itemize}
  \item $\Sc$ schedules $\O$; $\O$ sends $r_1^\O$ to $\ttp$;
  \item $\Sc$ schedules $\ttp$; $\ttp$ resolves the protocol for $\O$
    and sends $r_2^\O$;
  \item $\Sc$ schedules $\O$; $\O$ aborts the protocol by sending
    $a_1^\O$;
  \item $\Sc$ schedules $\ttp$; $\ttp$ sends $\comp{a_2^\O, a_2^\R}$ 
    with $\R$ having no option of obtaining $\O$'s signature;
  \item $\varphi_\R$ and $\varphi_\ttp$ are violated but $\varphi_\O$
    is satisfied.
  \end{itemize}
\item R3: $\seq{m_1, r_1^\R}$
  \begin{itemize}
  \item $\Sc$ schedules $\ttp$; $\ttp$ resolves and sends 
    $\comp{r_2^\O, r_2^\R}$;
  \item $\Sc$ schedules $\O$; $\O$ sends $a_1^\O$ to $\ttp$;
  \item $\Sc$ schedules $\ttp$; $\ttp$ sends $\comp{a_2^\O, a_2^\R}$;
  \item $\varphi_\R$, $\varphi_\ttp$ and $\varphi_\O$
    are violated.
  \end{itemize}
\item R4: $\seq{m_1, m_2, r_1^\R}$
  \begin{itemize}
  \item $\Sc$ schedules $\ttp$; $\ttp$ resolves and sends 
    $\comp{r_2^\O, r_2^\R}$;
  \item $\Sc$ schedules $\O$; $\O$ sends $a_1^\O$ to $\ttp$;
  \item $\Sc$ schedules $\ttp$; $\ttp$ sends $a_2^\R$;
  \item $\varphi_\R$ and $\varphi_\ttp$ are violated but $\varphi_\O$
    is satisfied.
  \end{itemize}
\end{enumerate}
\qed
\end{proof}

It is easy to verify that the sequences in the proof are exhaustive.
From the agent interface automaton shown in Figure~\ref{figure:agents}
we can extract all the partial sequences of moves ending in a move
of $\R$ and similarly for $\O$.
In all of the above cases, $\varphi_\R$ is violated.
In all of the above cases $\varphi_\R \wedge \varphi_\ttp$ is
also violated.
This shows that for all counter moves of $\O$ and the $\ttp$, 
violation of the specification of $\R$ also violates the 
specification of $\O$ or the $\ttp$.
Since $\O$ and the $\ttp$ co-operate, $\O$ never sends $m_3$,
instead choosing to use the $\ttp$ to get her non-repudiation 
evidence while denying $\R$ the ability to get his evidence.

The following example illustrates that given our objectives,
given a reasonable $\ttp$ as defined in Section~\ref{sec:fair-nonrep},
weak co-synthesis yields solutions that are not attack-free and
are hence unacceptable.

\begin{example}{\bf{(Weak co-synthesis generates unacceptable 
solutions)}}
Consider a refinement $\O'$, $\R'$ and $\ttp'$, that generates
the following sequence of messages:
$\seq{m_1, m_2, r_1^\O, r_2^\O, r_1^\R, r_2^\R}$; the agents
send $m_1$ and $m_2$ and then resolve the protocol individually.
We assume that $\ttp'$ needs both $m_1$ and $m_2$ to resolve the
protocol for either $\O$ or $\R$.
The trace corresponding to this sequence satisfies weak 
co-synthesis, but then this behavior of the $\ttp$, that assumes 
co-operative agent behavior, is not attack-free.
Taking $Y = \set{\R}$, consider the following $Y$-attack
where $\R$ exploits the fact that a reasonable 
$\ttp$ responds with $r_2^\R$ when she receives $r_1^\R$.
If $\R$ sends a resolve request immediately after receiving $m_1$
we get the message sequence $\seq{m_1, r_1^\R, r_2^\R}$.
In this case $\varphi_\R$ is satisfied, but $\varphi_\O$ and 
$\varphi_\ttp$ are violated.
The only way to satisfy $\varphi_\O$ and $\varphi_\ttp$ is if 
$\O'$ sends $r_1^\O$, which she cannot do, as she does not know
the contents of $m_2$.
This is an attack on the ASW certified mail protocol 
that compromises fairness for $\O$ \cite{KremerR03}.
Similarly, there exists a $Y$-attack for $Y = \set{\O, \R}$ as follows:
after resolving the protocol, if $\O$ decides to send $m_3$ and
$\R$ responds with $m_4$, we get the following
message sequence: $\seq{m_1, m_2, r_1^\O, r_2^\O, m_3, m_4}$.
In this case, the objectives $\varphi_\O$ and 
$\varphi_\R$ are satisfied but the objective $\varphi_\ttp$ 
is violated; a reasonable $\ttp$ will only send messages in
response to abort and resolve requests and thus needs $r_1^\R$ 
to satisfy $\varphi_\ttp$.
Therefore, solutions that satisfy weak co-synthesis may not be 
attack-free.
\qed
\end{example}

\subsection{The Need for a $\ttp$}

We now provide a justification of the need for a $\ttp$ in fair 
non-repudiation protocols, given our synthesis objective.
While this follows from \cite{Yacobi80,Pagnia99}, our proof gives
an alternative game-theoretic proof through synthesis.
We present the following theorem which shows that if we remove 
the $\ttp$, then both classical and assume-guarantee synthesis fail to 
synthesize a fair non-repudiation protocol.

\begin{theo}{\ (Classical and assume-guarantee synthesis fail without the $\ttp$)} 
\label{theo-no-ttp}
For all refinements $\O' \preceq \O$, the following assertions hold:
\begin{enumerate}
\item {\em Classical co-synthesis fails\/}:
$\semb{\O' \para \R \para \Sc} \not \subseteq \varphi_\O$.
\item {\em Assume-guarantee synthesis fails\/}:
\begin{enumerate}
\item $\semb{\O' \para \R \para \Sc} \not \subseteq 
(\varphi_\R \im \varphi_\O)$ or,
\item
$\semb{\O' \para \R \para \Sc} \subseteq 
(\varphi_\R \im \varphi_\O)$;
$\semb{\R' \para \O \para \Sc} \subseteq 
(\varphi_\O \im \varphi_\R)$; and \\
$\semb{\O' \para \R' \para \Sc} \not \subseteq (\varphi_\O \wedge
\varphi_\R)$.
\end{enumerate}
\end{enumerate}
\end{theo}

\begin{proof}
We note that as the $\ttp$ is not involved, $\abo$, $\abr$, $\siga^\ttp$
and $\sigb^\ttp$ are always false.
The agent objectives then simplify to,
\begin{align*}
\varphi_\O = \diam M_1 \wedge \diam \sigb^\R; \qquad \qquad
\varphi_\R = \bo (\eoo \im \diam \siga^\O) \eqpun .
\end{align*}
For assertion 1, consider an arbitrary refinement $\O' \preceq \O$.
We show a witness trace in 
$\semb{\O' \para \R \para \Sc}$ that violates $\varphi_\O$.
If $\O'$ does not send $m_1$ in the initial protocol state $v_0$, 
then we have a witness trace that trivially violates $\varphi_\O$ and
hence $\semb{\O' \para \R \para \Sc} \not \subseteq \varphi_\O$.
Assume $\O'$ sends $m_1$ and the channel between $\O$ and $\R$ eventually 
delivers all messages.
Consider a partial trace ending in protocol state 
$\set{M_1, \eoo, M_2, \eor}$;
messages $m_1$ and $m_2$ have been received.
The only choice of moves for $\O'$ in this state of the protocol
are $\idle$ or $m_3$.
If $\O'$ chooses $\idle$, then the trace does not satisfy 
$\varphi_\O$ and hence
$\semb{\O' \para \R \para \Sc} \not \subseteq \varphi_\O$.
If $\O'$ chooses $m_3$ and upon receiving $m_3$ if $\R$ decides
to stop participating in the protocol by choosing $\idle$, then
the trace satisfies $\varphi_\R$ but violates $\varphi_\O$ and hence
$\semb{\O' \para \R \para \Sc} \not\subseteq \varphi_\O$.

For assertion 2, consider an arbitrary refinement $\O' \preceq \O$.
If $\O'$ does not send $m_1$ in the initial protocol state $v_0$,
we have a witness trace that trivially violates $\varphi_\O$ but 
satisfies $\varphi_\R$.
Therefore, the trace does not satisfy $\varphi_\R \im \varphi_\O$
and
$\semb{\O' \para \R \para \Sc} \not\subseteq (\varphi_\R \im 
\varphi_\O)$.
Assume the channels eventually deliver all messages and as in
the proof of assertion 1, consider a partial trace ending in protocol 
state $\set{M_1, \eoo, M_2, \eor}$.
To produce a witness trace we have the following cases based on
the move chosen by $\O'$:
\begin{itemize}
\item {\em Case 1. $\O'$ chooses $\idle$.\/}
Since $\O'$ chooses $\idle$, she does not send her signature
$\siga^\O$.
Therefore, the trace does not satisfy $\varphi_\R$.
Since $\R$ sends $m_4$ only in response to $m_3$, $\O$ does
not get $\sigb^\R$ from $\R$ in this case.
Therefore, the trace does not satisfy $\varphi_\O$ either and hence 
satisfies $\varphi_\O \im \varphi_\R$ and $\varphi_\R \im \varphi_\O$ 
but does not satisfy $\varphi_\O \wedge \varphi_\R$.
This leads to,
$\semb{\O' \para \R \para \Sc} \subseteq (\varphi_\O \im \varphi_\R)$
and 
$\semb{\O' \para \R \para \Sc} \subseteq (\varphi_\R \im \varphi_\O)$
but 
$\semb{\O' \para \R \para \Sc} \not\subseteq (\varphi_\O \wedge 
\varphi_\R)$
\item {\em Case 2. $\O'$ chooses $m_3$.\/}
Since $m_3$ is eventually delivered, $\R$ gets his 
non-repudiation evidence and the trace satisfies $\varphi_\R$.
If $\R$ now stops participating in the protocol and chooses the
idle move $\idle$ instead of sending $m_4$, then $\O$ does not
get her non-repudiation evidence and the trace does not satisfy
$\varphi_\O$.
We therefore have a witness trace that does not satisfy
$\varphi_\R \im \varphi_\O$.
This leads to,
$\semb{\O' \para \R \para \Sc} \not\subseteq (\varphi_\R \im 
\varphi_\O)$
\end{itemize}
\qed
\end{proof}

If the agents co-operate, then a refinement $\O' \preceq \O$ that
sends $m_1$ and then $m_3$ upon receiving $m_2$ and similarly
a refinement $\R' \preceq \R$ that sends $m_2$ and $m_4$ upon 
receiving $m_1$ and $m_3$ respectively, is a solution to the
weak co-synthesis problem.
The sequence of messages in this case is precisely $\seq{m_1, m_2,
m_3, m_4}$ which is the main protocol in all the fair exchange
protocols we have studied.
The problem arises when either $\O$ or $\R$ are dishonest and
try to cheat the other agent.

\subsection{Assume-guarantee Solutions are Attack-Free}

In this subsection we show that assume-guarantee solutions are
attack free; no coalition of participants can violate the objective
of at least one of the other participants while satisfying their
own objectives.
Let $P' = (\O', \R', \ttp')$ be a tuple of refinements of the agents
and the $\ttp$.
For two refinements $P' = (\O', \R', \ttp')$ and 
$P'' = (\O'', \R'', \ttp'')$, we write $P' \preceq P''$ if 
$\O' \preceq \O''$, $\R' \preceq \R''$ and $\ttp' \preceq \ttp''$.
Given $P = (\O, \R, \ttp)$, the most general behaviors of
the agents and the $\ttp$, let $P_{AGS}$ be the set of all possible 
refinements $P' \preceq P$ that satisfy the conditions of 
assume-guarantee synthesis. 
For a refinement $P' = (\O', \R', \ttp')$ to be in $P_{AGS}$, we 
require that the refinements $\O' \preceq \O$, $\R' \preceq \R$ and 
$\ttp' \preceq \ttp$ satisfy the following conditions:

For all fair schedulers $\Sc$, for all possible behaviors of
the channels,
\begin{enumerate}
\item $\semb{(\O' \para \R \para \ttp \para \Sc)}
\subseteq (\varphi_\R \wedge \varphi_\ttp) \im \varphi_\O$;
\item $\semb{(\O \para \R' \para \ttp \para \Sc)}
\subseteq (\varphi_\O \wedge \varphi_\ttp) \im \varphi_\R$;
\item $\semb{(\O \para \R \para \ttp' \para \Sc)}
\subseteq (\varphi_\O \wedge \varphi_\R) \im \varphi_\ttp$;
\item $\semb{(\O' \para \R' \para \ttp' \para \Sc)} 
\subseteq (\varphi_\O \wedge \varphi_\R \wedge \varphi_\ttp)$.
\end{enumerate}
We now characterize the smallest restriction on the refinements
$\ttp' \preceq \ttp$ that satisfy the implication condition,
\begin{align}
\semb{(\O \para \R \para \ttp' \para \Sc)}
\subseteq (\varphi_\O \wedge \varphi_\R) \im \varphi_\ttp
\eqpun .
\label{ttp-implication}
\end{align}
In order to characterize the smallest restriction on $\ttp'$ we first
define the following constraints on the $\ttp$ and prove 
that they are both necessary and sufficient to satisfy 
(\ref{ttp-implication}).

\medskip\noindent{\bf AGS constraints on the $\ttp$.}
We say that a refinement $\ttp' \preceq \ttp$ satisfies the 
{\em AGS constraints on the $\ttp$}, if $\ttp'$ satisfies the
the following constraints:
\begin{enumerate}
\item {\em Abort constraint.\/}
If the first request received by the $\ttp$ is an abort
request, then her response to that request should be 
$\comp{a_2^\O, a_2^\R}$;
\item {\em Resolve constraint.\/}
If the first request received by the $\ttp$ is a resolve
request, then her response to that request should be 
$\comp{r_2^\O, r_2^\R}$;
\item {\em Accountability constraint.\/}
If the first response from the $\ttp$ is 
$\comp{x, y}$, then for all subsequent abort or resolve requests
her response should be in the set $\set{\idle, x, y, \comp{x, y}}$.
\end{enumerate}
We assume a reasonable $\ttp$, as defined in 
Section~\ref{sec:fair-nonrep}; 
in particular she only responds to abort or resolve requests.
In the following lemma, in assertion 1 we show that for all 
refinements $\ttp' \preceq \ttp$ that satisfy the AGS constraints on 
the $\ttp$, we have $\ttp'$ is inviolable , i.e., neither agent can
violate the objective $\varphi_\ttp$, and hence satisfies the
implication condition (\ref{ttp-implication}); in assertion (2) we
show that if $\ttp'$ does not satisfy the AGS constraints on the
$\ttp$, the implication condition (\ref{ttp-implication}) is not 
satisfied.

\begin{lem}{} \label{lem-ttp-inviol}
For all refinements $\ttp' \preceq \ttp$, the following assertions
hold:
\begin{enumerate}
\item if $\ttp'$ satisfies the AGS constraints on the $\ttp$, then
\[
\semb{\O \para \R \para \ttp' \para \Sc} \subseteq \varphi_\ttp
\subseteq (\varphi_\O \wedge \varphi_\R) \im \varphi_\ttp.
\]
\item if $\ttp'$ does not satisfy the AGS constraints on the $\ttp$,
then
\[
\semb{\O \para \R \para \ttp' \para \Sc} \not\subseteq
(\varphi_\O \wedge \varphi_\R) \im \varphi_\ttp.
\]
\end{enumerate}
\end{lem}
\begin{proof}
For assertion 1, consider an arbitrary $\ttp' \preceq \ttp$
that satisfies the AGS constraints on the $\ttp$.
We consider the following cases of sets of traces of
$\semb{\O \para \R \para \ttp' \para \Sc}$ for the proof:
\begin{itemize}
\item {\em Case 1. Neither agent aborts nor resolves the protocol.\/}
In these traces, since the $\ttp$ is neither sent an abort nor a
resolve request, $\varphi_\ttp$ is satisfied trivially.
Therefore, all these traces satisfy 
$(\varphi_\O \wedge \varphi_\R) \im \varphi_\ttp$.
\item {\em Case 2. The first request to the $\ttp$ is an abort
request.\/}
For the set of traces where the first request to the $\ttp$ is
an abort request, given $\ttp'$ satisfies the AGS constraints on the 
$\ttp$, by the abort constraint, the response of the $\ttp$ to this 
request is $\comp{a_2^\O, a_2^\R}$.
For all subsequent abort or resolve requests, by the
accountability constraint, the $\ttp$ responds with a move in 
set $\set{\idle, a_2^\O, a_2^\R, \comp{a_2^\O, a_2^\R}}$.
This implies that both agents get the abort token and neither agent
gets non-repudiation evidences.
Therefore, $\varphi_\ttp$ is satisfied for all these traces and 
hence
$(\varphi_\O \wedge \varphi_\R) \im \varphi_\ttp$ is also satisfied.
\item {\em Case 3. The first request to the $\ttp$ is a resolve
request.\/}
Similar to the proof of Case 2, in the set of traces where the
first request to the $\ttp$ is a resolve request, by the resolve
constraint, the $\ttp$ responds to this request with move 
$\comp{r_2^\O, r_2^\R}$.
Since the response of the $\ttp$ to all subsequent abort or
resolve requests is in the set $\set{\idle, r_2^\O, r_2^\R, 
\comp{r_2^\O, r_2^\R}}$, by the accountability constraint, the agents 
get their non-repudiation evidences and neither gets the abort token.
Therefore, $\varphi_\ttp$ is satisfied for all these traces and hence
$(\varphi_\O \wedge \varphi_\R) \im \varphi_\ttp$ is also satisfied
and the result follows.
\end{itemize}

For assertion 2, consider an arbitrary $\ttp' \preceq \ttp$
that does not satisfy the AGS constraints on the $\ttp$.
We assume a reasonable $\ttp$ and consider violation of the AGS 
constraints on the $\ttp$ on a case by case basis.
For each case we produce a witness trace that violates the 
implication condition 
$(\varphi_\O \wedge \varphi_\R) \im \varphi_\ttp$.
We proceed as follows:
\begin{itemize}
\item {\em Case 1. The abort constraint is violated.\/}
To produce a witness trace we consider a partial trace that ends
in protocol state $\set{M_1, \reqa^\O}$; $\O$ requests the $\ttp$ to
abort the protocol after sending message $m_1$ but before it is
received.
Since $\ttp'$ violates the abort constraint, the only choice of 
moves for $\ttp'$ are $\idle$ or $a_2^\O$.
This leads to the following cases:
\begin{itemize}
\item {\em Case (a). $\ttp'$ chooses $\idle$.\/}
It is trivially the case that $\varphi_\ttp$ is violated for this
trace as $\varphi_\ttp^1$ is violated.
At this stage in the protocol, there exists a behavior of $\O$,
$\R$ and the channel between $\O$ and $\R$, where the channel
delivers all messages and the agents co-operate and complete the
protocol by exchanging their signatures.
Therefore, $\varphi_\O \wedge \varphi_\R$ is satisfied but
$\varphi_\ttp$ is violated.
Therefore, the trace does not satisfy 
$(\varphi_\O \wedge \varphi_\R) \im \varphi_\ttp$.
\item {\em Case (b). $\ttp'$ chooses $a_2^\O$.\/}
Since the channel between the agents and the $\ttp$ is resilient,
$\O$ eventually receives her abort token $\abo$.
At this stage in the protocol, there exists a 
behavior of $\O$, $\R$ and the channel between $\O$ and $\R$ such
that the channel delivers all messages and the agents exchange
their signatures, leading to the satisfaction of 
$\varphi_\O \wedge \varphi_\R$ but a violation of 
$\varphi_\ttp^4$ and hence $\varphi_\ttp$.
Therefore, the trace does not satisfy 
$(\varphi_\O \wedge \varphi_\R) \im \varphi_\ttp$.
\end{itemize}
\item {\em Case 2. The resolve constraint is violated.\/}
To produce a witness trace we consider a partial trace that ends
in protocol state $\set{M_1, \eoo, M_2, \eor, \reqr^\O}$; $\O$
resolves the protocol after messages $m_1$ and $m_2$ have been 
received.
Since $\ttp'$ violates the resolve constraint, the only choice of
moves for $\ttp'$ are $\idle$ or $r_2^\O$.
An argument similar to the argument for cases 1(a) and 1(b) again
leads to the satisfaction of $\varphi_\O \wedge \varphi_\R$ but
a violation of $\varphi_\ttp$.
\item {\em Case 3. The accountability constraint is violated.\/}
To produce a witness trace we consider a partial trace that ends
in protocol state $\set{M_1, \eoo, M_2, \eor, \reqa^\O, \reqr^\R,
A_2^\O, A_2^\R, \abo, \abr}$; 
$\O$ aborts the protocol and $\R$ resolves the protocol 
after messages $m_1$ and $m_2$ have been received.
The $\ttp$ receives the abort request before the resolve request
and aborts the protocol by sending $\comp{a_2^\O, a_2^\R}$.
Since $\ttp'$ violates the accountability constraint, the only 
choice of moves for $\ttp'$ to the resolve request from $\R$ are 
$r_2^\R$ or $\comp{r_2^\O, r_2^\R}$.
The leads to the following cases:
\begin{itemize}
\item {\em Case (a). $\ttp'$ chooses $r_2^\R$.\/}
This violates $\varphi_\ttp^4$ and $\varphi_\ttp^5$ and hence
violates $\varphi_\ttp$.
At this stage in the protocol, there exists a behavior of $\O$,
$\R$ and the channel between $\O$ and $\R$ such that the agents
exchange their signatures and complete the protocol thus satisfying 
$\varphi_\O \wedge \varphi_\R$.
Therefore, this trace does not satisfy the implication condition
$(\varphi_\O \wedge \varphi_\R) \im \varphi_\ttp$.
\item {\em Case (b). $\ttp'$ chooses $\comp{r_2^\O, r_2^\R}$.\/}
This violates $\varphi_\ttp^4$ and $\varphi_\ttp^5$ and hence
violates $\varphi_\ttp$.
An argument similar to Case 2(a) leads to a violation of 
$(\varphi_\O \wedge \varphi_\R) \im \varphi_\ttp$ for this trace.
\end{itemize}
As we have shown witness traces that do not satisfy the implication
condition $(\varphi_\O \wedge \varphi_\R) \im \varphi_\ttp$ when
$\ttp'$ violates any of the AGS constraints on the $\ttp$, the
result follows.
\end{itemize}
\qed
\end{proof}

In the following theorem we show that all refinements $P' \in
P_{AGS}$ are attack-free; no subset of participants can violate the
objective of at least one of the other participants while satisfying
their own objectives.

\begin{theo}{}
\label{theo-ags-attack-free}
All refinements $P' \in P_{AGS}$ are attack-free.
\end{theo}

\begin{proof}
We show that for all refinements $P' \in P_{AGS}$ there exists
no $Y$-attack for all $Y \subseteq \set{\O, \R, \ttp}$.
Let $P' = (\O', \R', \ttp')$ and $A = \set{\O, \R, \ttp}$ be the
set of participants.
We have the following cases:
\begin{itemize}
\item {\em Case 1. $|Y| = 0$.\/}
In this case $Y = \emptyset$ and 
$(A \setm Y)' = \set{\O', \R', \ttp'}$.
Since $(A \setm Y)'$ are the refinements in $P'$ which is in $P_{AGS}$, 
by the weak co-synthesis condition, the objectives
$\varphi_\O$, $\varphi_\R$ and $\varphi_\ttp$ are satisfied.
Therefore there is no $Y$-attack in this case.
\item {\em Case 2. $|Y| = 1$.\/}
We first show that there is no $Y$-attack for $Y = \set{\O}$.
The case of $Y = \set{\R}$ is similar.
By Lemma~\ref{lem-ttp-inviol} (assertion 2), for all refinements 
$P' \in P_{AGS}$, the refinement $\ttp'$ must satisfy the AGS 
constraints on the $\ttp$.
This implies, by Lemma~\ref{lem-ttp-inviol} (assertion 1), neither 
$\O$ nor $\R$ can violate $\varphi_\ttp$.
Since $\varphi_\ttp$ cannot be violated, a $Y$-attack in this
case must generate a trace where $\varphi_\R$ is violated but 
$\varphi_\O$ is satisfied.
But this violates the implication condition,
$\varphi_\O \wedge \varphi_\ttp \im \varphi_\R$,
contradicting the assumption that $P' \in P_{AGS}$.
We now show that there is no $Y$-attack for $Y = \set{\ttp}$.
Since we assume the $\ttp$ is reasonable, in all traces where 
neither agent sends an abort nor a resolve request to the $\ttp$, 
the $\ttp$ cannot violate the agent objectives.
In all traces where the first request from the agents is an abort 
request, given a reasonable $\ttp$, since the trace satisfies
$\varphi_\ttp$, it must be the case that the response to that request 
is $\comp{a_2^\O, a_2^\R}$.
Similarly, for resolve requests.
If the first response of the $\ttp$ is $\comp{x, y}$, then the only 
responses that satisfy $\varphi_\ttp$, to all subsequent abort and
resolve requests, are in the set $\set{\idle, x, y, \comp{x, y}}$.
This implies that either the agents get abort tokens or 
non-repudiation evidences but never both, which implies 
$\varphi_\O$ and $\varphi_\R$ are satisfied in all these traces.
Therefore there is no $Y$-attack in this case as well.
\item {\em Case 3. $|Y| = 2$.\/}
Since $P' \in P_{AGS}$, by the implication conditions of 
assume-guarantee synthesis, there cannot be a $Y$-attack where
$|Y| = 2$.
\item {\em Case 4. $|Y| = 3$.\/}
It is trivially the case that there is no $Y$-attack as 
$(A \setm Y)' = \emptyset$.
\end{itemize}
Since we have shown that for all refinements $P' \in P_{AGS}$,
for all $Y \subseteq A$, there is no $Y$-attack in $P'$, we conclude
that all refinements in $P_{AGS}$ are attack-free.
\qed
\end{proof}

We now present the following theorem that establishes conditions 
for any refinement in $P_{AGS}$ to be an attack-free fair 
non-repudiation protocol.

\begin{theo}{ (Fair non-repudiation protocols)} 
\label{theo-fair-nonrep}
For all refinements $P' \in P_{AGS}$, if 
$\semb{\O' \para \R' \para \ttp' \para \Sc} \cap
(\diam \nro \wedge \diam \nrr) \ne \emptyset$, then
$P'$ is an attack-free fair non-repudiation protocol.
\end{theo}

\begin{proof}
Consider an arbitrary refinement $P' = (\O', \R', \ttp') \in P_{AGS}$.
Since $P' \in P_{AGS}$, by Theorem~\ref{theo-ags-attack-free}, 
it is attack-free.
Further, by the weak co-synthesis condition, we have
$\semb{\O' \para \R' \para \ttp' \para \Sc} \subseteq 
(\varphi_\O \wedge \varphi_\R \wedge \varphi_\ttp)$ and hence
by Theorem~\ref{objectives-imply-ft},
we have $\semb{\O' \para \R' \para \ttp' \para \Sc} \subseteq 
\varphi_f$.
Thus $P'$ satisfies fairness.
Using PCS, that provides the designated verifier property, to 
encrypt all messages, we ensure that the protocol is abuse-free.
Since $\semb{\O' \para \R' \para \ttp' \para \Sc} \cap
(\diam \nro \wedge \diam \nrr) \ne \emptyset$, the refinement $P'$
enables an exchange of signatures and hence is an exchange protocol.
Given $\nro$ and $\nrr$ are non-repudiation evidences for $R$ and $\O$
respectively, we conclude that $P'$ is an attack-free fair 
non-repudiation protocol.
\qed
\end{proof}

%% file: analysis.tex
\subsection{Analysis of Existing Fair Non-repudiation Protocols 
as $P_{AGS}$ Solutions}

In this subsection we analyze existing fair non-repudiation 
protocols and check if they are solutions to assume-guarantee
synthesis.
To facilitate the analysis, we first present an alternate 
characterization of the set $P_{AGS}$ of assume-guarantee 
refinements.
We then show that the KM non-repudiation protocol with offline $\ttp$ 
is in $P_{AGS}$ whereas the ASW certified mail protocol and the GJM 
protocol are not.
Finally, we present a systematic exploration of refinements leading
to the KM protocol.
Towards an alternate characterization of $P_{AGS}$, we begin
by defining constraints on $\O$, similar to the AGS constraints on 
the $\ttp$ that ensure satisfaction of the implication condition
for $\O$.
We then define maximal and minimal refinements that satisfy all the
implication conditions of assume-guarantee synthesis and introduce a
{\em bounded idle time\/} requirement to ensure satisfaction
of weak co-synthesis.

\medskip\noindent{\bf AGS constraints on $\O$.}
Given $P = (\O, \R, \ttp)$, the most general behaviors of the
agents and the $\ttp$, we say a refinement $P' \preceq P$ satisfies
the {\em AGS constraints on $\O$\/}, if the following
conditions hold:
\begin{enumerate}
\item $a_1^\O \not\in \mov_{\O'}(v_0)$;
\item $\siga^\O \not\in \mov_{\O'}(\set{M_1, \eor, \reqa^\O})$; and
\item $a_1^\O \not\in \mov_{\O'}(\set{M_1, \eor, M_3})$.
\end{enumerate}
In the Appendix, we show that these constraints are both necessary 
and sufficient restrictions on the moves of $\O$ that satisfy 
the implication condition
$(\varphi_\R \wedge \varphi_\ttp) \im \varphi_\O$ of assume-guarantee
synthesis.
We also show that all refinements $\R' \preceq \R$ satisfy the
implication condition
$(\varphi_\O \wedge \varphi_\ttp) \im \varphi_\R$ of assume-guarantee
synthesis.

\medskip\noindent{\bf The maximal refinement $\pmax$.}
We define the maximal refinement 
$\pmax = (\omax, \rmax, \tmax)$ as follows:
\begin{enumerate}
\item $\omax \preceq \O$ satisfies the AGS constraints on $\O$
and for all $\O'$ that satisfy the constraints, we have
$O' \preceq \omax$;
\item $\rmax = \R$; and
\item $\tmax \preceq \ttp$ satisfies the AGS constraints on
the $\ttp$ and for all $\ttp'$ that satisfy the constraints,
we have $\ttp' \preceq \tmax$.
\end{enumerate}

We show in the Appendix the correspondence between $\pmax$ and
the smallest restriction on the moves of $\O$ and the $\ttp$
so that $\pmax$ is a witness to $P_{AGS}$.
While there are restrictions on $\O$ and the $\ttp$, there are
no restrictions on $\R$.

\medskip\noindent{\bf The minimal refinement $\pmin$.}
We present the smallest refinement $\pmin = (\omin, \rmin, \tmin)$ 
in $P_{AGS}$, as the largest restriction on the
moves of $\O$, $\R$ and the $\ttp$, as follows:
\begin{enumerate}
\item $\pmin \preceq \pmax$; 
\item $\moves_{\omin} = \set{m_1, a_1^\O}$;
\item $\moves_{\rmin} = \set{\idle}$;
\item $\omin$ satisfies the AGS constraints on $\O$; and
\item $\tmin$ satisfies the AGS constraints on the $\ttp$.
\end{enumerate}
If $m_1 \not\in \moves_{\omin}$, then $\varphi_\O$ cannot be
satisfied as $\omin$ does have the ability to initiate a protocol
instance.
If $a_1^\O \not\in \moves_{\omin}$, then $\varphi_\O$ cannot be
satisfied whether or not $m_1$ is delivered, as $\rmin$ has no choice of
moves other than $\idle$.
If $\omin$ does not satisfy the AGS constraints on $\O$ and sends
$a_1^\O$ in the initial state of the protocol $v_0$, then the
resulting trace trivially violates $\varphi_\O$ while satisfying
$\varphi_\R \wedge \varphi_\ttp$.

\medskip\noindent{\bf The bounded idle time requirement.\/}
We say that a refinement $P'$ satisfies {\em bounded idle time\/}
if $\O$ and the $\ttp$ in $P'$ choose the idle move $\idle$, 
when scheduled by $\Sc$, at most $b$ times for a finite $b \in \nats$.
We prove that satisfaction of the bounded idle time requirement is both 
necessary and sufficient to ensure satisfaction of the weak co-synthesis 
condition of assume-guarantee synthesis, for all refinements that
satisfy the AGS constraints on the $\ttp$ and the AGS constraints
on $\O$, in the Appendix.

\medskip\noindent{\bf Alternate characterization of $P_{AGS}$.}
We now use $\pmin$ and $\pmax$ to provide an alternate 
characterization of the set $P_{AGS}$.
We first define the following set of refinements $\pb$:
\begin{align*}
\pb = \set{& P' = (\O', \R', \ttp') \ \mid \ P' 
\text{ satisfies bounded idle time}; \pmin \preceq P' \preceq \pmax;\\
& \ttp' \text{ satisfies the AGS constraints on the \ttp}}
\eqpun .
\end{align*}
The following lemma states that the set $\pb$ and the set $P_{AGS}$
coincide.
We present the lemma here and prove it in the Appendix.

\begin{lem}{\ (Alternate characterization of $P_{AGS}$)}
\label{lem-alt}
We have $\pb = P_{AGS}$.
\end{lem}

\medskip\noindent{\bf The KM non-repudiation protocol.}
The KM protocol, like the ASW and GJM protocols 
consists of a main protocol, an abort subprotocol and a resolve
subprotocol.
The main protocol is the same as in the ASW and GJM protocols and
is defined in terms of messages in Protocol~\ref{prot:main}.
The abort subprotocol and the resolve subprotocol are defined
in Table~\ref{table-ttp-moves}.
Let $P_{KM} = (\O_{KM}, \R_{KM}, \ttp_{KM})$ correspond to the
agent and $\ttp$ refinements in the KM protocol.
Since $\O$ does not abort the protocol in
state $v_0$ and in state $\set{M_1, \eor, M_3}$ in $\O_{KM}$,
it follows that $\omin \preceq \O_{KM} \preceq \omax$.
It is easy to verify that $\rmin \preceq \R_{KM} \preceq
\rmax$ and $\tmin \preceq \ttp_{KM} \preceq \tmax$.
Moreover, $\ttp_{KM}$ satisfies the AGS constraints
on the $\ttp$ and $P_{KM}$ satisfies bounded idle time.
Therefore $P_{KM} \in \pb$ and hence by Lemma~\ref{lem-alt}, 
$P_{KM} \in P_{AGS}$.
\IncMargin{0.25em}
\SetAlgorithmName{Protocol}{Algorithm}{}

\begin{algorithm}[t]
  \caption{$\textsc{The KM, ASW and GJM Main Protocol}$}
  \label{prot:main}
    $\O$ sends $m_1$ to $\R$\;
    $\R$ sends $m_2$ to $\O$\;
    \eIf{($\R$ does not send $m_2$ on time)}{
      $\O$ sends $a_1^\O$ to the $\ttp$\;
    }{
      $\O$ sends $m_3$ to $\R$\;
      \eIf{($\O$ does not send $m_3$ on time)}{
        $\R$ sends $r_1^\R$ to the $\ttp$\;
      }{
        $\R$ sends $m_4$ to $\O$\;
        \If{($\R$ does not send $m_4$ on time)}{
          $\O$ sends $r_1^\O$ to the $\ttp$\;
        }
      }
    }
\end{algorithm}

\medskip\noindent{\bf The ASW certified mail protocol.}
The ASW certified mail protocol differs from the KM protocol in its 
abort and resolve sequences.
To define the abort protocol, the $\ttp$ needs a move $req^O$ that
can be used to request $\O$ to resolve a protocol instance if $\R$ 
has already resolved it.
The abort and resolve subprotocols are defined in
Table~\ref{table-ttp-moves}.
Let $P_{ASW} = (\O_{ASW}, \R_{ASW}, \ttp_{ASW})$ correspond to the
agent and $\ttp$ refinements in the ASW certified mail protocol.
Since $\ttp_{ASW}$ neither has move $\comp{a_2^\O, a_2^\R}$ nor
$\comp{r_2^\O, r_2^\R}$, $\ttp_{ASW}$ does not satisfy the AGS 
constraints on the $\ttp$ and hence by Lemma~\ref{lem-ttp-inviol}
(assertion 2), we have $P_{ASW} \not\in P_{AGS}$.
Moreover, the ASW certified mail protocol is not attack-free as shown 
by the following attacks \cite{KremerR03}:
Consider a behavior of the channels that deliver all messages and
the sequence of messages $\seq{m_1, r_1^\R, r_2^\R, a_1^\O, 
req^O}$.
This is a valid sequence in the ASW protocol.
In this sequence a malicious $\R$ decides to resolve the protocol 
after receiving $m_1$ and thus succeeds in getting $\siga^\ttp$.
When $\O_{ASW}$ attempts to abort the protocol, $\ttp_{ASW}$ expects her
to resolve the protocol as $\R$ has already resolved it, but $\O_{ASW}$
cannot do so as she does not have $m_2$.
Therefore, $\varphi_\O$ is violated; $\O_{ASW}$ cannot abort or resolve
the protocol, neither can she get $\R$'s signature.
Consider the sequence of messages 
$\seq{m_1, m_2, r_1^\O, r_2^\O, a_1^\O, a_2^\O}$.
This is an attack that compromises fairness for $\R$; in the words of 
\cite{KremerR03} the protocol designers did not foresee that $\O$ could 
resolve the protocol and then abort it.
This violates $\varphi_\R$ and $\ttp$ accountability, violating 
$\varphi_\ttp$, while satisfying $\varphi_\O$.

\medskip\noindent{\bf The GJM protocol.}
The GJM protocol differs in the abort and resolve sequences
as shown in Table~\ref{table-ttp-moves}.
Garay et al., introduced the notion of abuse-freeness 
and invented {\em private contract signatures or PCS\/}, a
cryptographic primitive that ensures abuse-freeness and optionally
$\ttp$ accountability \cite{GarayJM99}.
Further, the GJM protocol is faithful to the informal definition of
fairness in that, when a protocol instance is aborted, neither agent 
gets partial information that can be used to negotiate a contract 
with a third party.
This is ensured by the use of PCS which provides the
{\em designated verifier property\/}; only $\R$ can verify the
authenticity of a message signed by $\O$ and vice versa.
The use of PCS in addition to the fixes to the original protocol 
proposed in \cite{ShmatikovM02} ensure that the protocol is free 
from replay attacks, is fair and abuse-free.
Let $P_{GJM} = (\O_{GJM}, \R_{GJM}, \ttp_{GJM})$ correspond to the
agent and $\ttp$ refinements in the GJM protocol.
Since $\ttp_{GJM}$ neither has move $\comp{a_2^\O, a_2^\R}$ nor
$\comp{r_2^\O, r_2^\R}$, $\ttp_{GJM}$ does not satisfy the AGS 
constraints on the $\ttp$ and hence by Lemma~\ref{lem-ttp-inviol}
(assertion 2), we have $P_{GJM} \not\in P_{AGS}$.
$P_{GJM}$ does not provide $\ttp$ inviolability and is not
attack-free by our definition.
Consider the message sequence
$g = \seq{m_1, m_2, m_3, r_1^\O, r_2^\O}$; agent $\R$
does not send his final signature but goes idle and stops participating
in the protocol after receiving $\O$'s signature.
$\O_{GJM}$ resolves the protocol by sending $r_1^\O$ and gets 
$\sigb^\ttp$.
In this case, while the objectives of $\O$ and $\R$ are satisfied, 
the $\ttp$ cannot satisfy $\varphi_\ttp$ unless $\R_{GJM}$ co-operates 
and sends a resolve request $r_1^\R$ after having satisfied his 
objective, which he may never do; it is rather unrealistic to expect
that he will.
Precisely,
$g \in \semb{\O \para \R \para \ttp_{GJM} \para \Sc}$ and
$g \not\in (\varphi_\O \wedge \varphi_\R) \im \varphi_\ttp$.

\begin{theo}{}
The refinement corresponding to the KM non-repudiation protocol is in
$P_{AGS}$ and the refinements corresponding to the ASW certified mail
protocol and the GJM protocol are not in $P_{AGS}$.
\end{theo}

\medskip\noindent{\bf Computation.}
We can obtain the solution of assume-guarantee synthesis by solving
graph games with {\em secure equilibria\/} \cite{CHJ06}.
In fact, the refinements that satisfy assume-guarantee synthesis
precisely correspond to secure equilibrium strategies of players
in the game.
This result was presented in \cite{CH07}.
All the objectives we consider in this paper are boolean 
combinations of B\"uchi ($\bo \diam$) and co-B\"uchi ($\diam \bo$)
objectives.
It follows from \cite{CH07} that secure equilibria with 
combinations of B\"uchi and co-B\"uchi objectives can be solved
in polynomial time.
This gives us a polynomial time algorithm for the assume-guarantee 
synthesis of fair exchange protocols.

\medskip\noindent{\bf From $P_{AGS}$ to $P_{KM}$.}
We now first present a systematic exploration of
the refinements of $P = (\O, \R, \ttp)$, the most general behavior
of the agents and the $\ttp$, leading to the KM protocol.
We consider the following refinements, that we assume satisfy
bounded idle time and the AGS constraints on the $\ttp$, and study 
their properties:
\begin{enumerate}
\item $\pmin = (\omin, \rmin, \tmin)$; the minimal refinement.
\item $P_1 = (\O_1, \R_1, \ttp_1)$ with \\
$\moves_{\O_1} = \moves_{\omin} \cup \set{\idle, m_3}$,
$\moves_{\R_1} = \moves_{\rmin} \cup \set{m_2, m_4}$ and 
$\ttp_1 = \tmax$.
\item $P_2 = (\O_2, \R_2, \ttp_2)$ with \\
$\moves_{\O_2} = \moves_{\O_1} \cup \set{r_1^\O}$,
$\moves_{\R_2} = \moves_{\R_1}$ and
$\ttp_2 = \tmax$.
\item $P_3 = (\O_3, \R_3, \ttp_3)$ with \\
$\moves_{\O_3} = \moves_{\O_2} \setm \set{a_1^\O}$,
$\moves_{\R_3} = \moves_{\R_1} \cup \set{r_1^\R}$ and
$\ttp_3 = \tmax$.
\item $\pmax = (\omax, \rmax, \tmax)$; the maximal refinement.
\end{enumerate}

\medskip\noindent{\bf Analysis of the refinement $\pmin$.}
It is easy to check that while $\pmin \in P_{AGS}$, it always
ends aborted as $a_1^\O$ is the only choice of moves for $\omin$
after $m_1$ is sent.
It is not an exchange protocol as it does not enable an exchange
of signatures.

\medskip\noindent{\bf Analysis of the refinement $\textrm{P}_1$.}
In this case, the agents do not have the ability to resolve the 
protocol.
The objectives of the agent and the $\ttp$ then reduce to,
\begin{gather*}
\varphi_\O = \diam M_1 \wedge \bo (\diam \sigb^\R \vee 
(\diam \abo \wedge \bo \neg \siga^\O)), \\
\varphi_\R = \bo (\eoo \im (\diam \siga^\O \vee
(\diam \abr \wedge \bo \neg \sigb^\R)), \\
\varphi_\ttp = \bo (\reqa \im (\diam \abo \vee \diam \abr)) 
\wedge \bo (\abo \im \diam \abr) \wedge 
\bo (\abr \im \diam \abo) \eqpun .
\end{gather*}
\begin{table*}[!t]
\begin{center}
\begin{tabular}{|l||l|l||l|l|l|}
\hline
\bf{Delivered message sequences \quad} & \multicolumn{5}{c|}{\bf{Moves for $\O_1$ and $\R_1$}} \\
\cline{2-6}
& \multicolumn{2}{c||}{\bf{Choices for $\O_1$}} & \multicolumn{3}{c|}{\bf{Choices for $\R_1$}} \\
\hline \hline
$\bf{\seq{}}$ & $m_1 \qquad$ & $m_1 \qquad$ & $\idle \qquad$ & $\idle \qquad$ & $\idle \qquad$ \\
$\bf{\seq{m_1}}$ & $\idle$ & $a_1^\O$ & $\idle$ & $m_2$ & \text{either} $\idle$ \text{or} $m_2$ \\
$\bf{\seq{m_1, m_2}}$ & $a_1^\O$ & $a_1^\O$ & $\idle$ & $\idle$ & $\idle$ \\
$\bf{\seq{m_1, m_2, m_3}}$ & $\emptyset$ & $\emptyset$ & $\idle$ & $m_4$ & $\idle$ \\
\hline
\end{tabular}
\end{center}
\caption[]{The moves that satisfy the objectives of assume-guarantee 
synthesis for $\O_1$ and $\R_1$ are shown in this table at relevant
protocol states represented by message sequences, when the agents have 
no ability to resolve the protocol.}
\label{table-stra-nor}
\end{table*}  
The agent moves that extend partial protocol runs such that the
implication conditions of assume-guarantee synthesis are satisfied 
in all resulting traces is shown in Table~\ref{table-stra-nor}.
Each row in the table corresponds to a protocol state and the moves 
available to $\O_1$ and $\R_1$ at that state, such that the 
implication conditions of assume-guarantee synthesis are satisfied
in all resulting traces.
For example, in the row corresponding to
$\seq{m_1}$, we have two move choices for $\O_1$, one 
that selects $\idle$ and the other that selects $a_1^\O$;
$\O_1$ can choose to wait for $\R$ to send $m_2$ or choose $a_1^\O$.
A similar interpretation is attached to the moves of $\R_1$.
We have $\pmin \preceq P_1 \preceq \pmax$.
As $P_1$ satisfies bounded idle time and the AGS constraints
on the $\ttp$, $P_1 \in \pb$ and hence, by Lemma~\ref{lem-alt},
$P_1 \in P_{AGS}$.
The refinement $P_1$, while attack-free, is not a fair non-repudiation 
protocol as it does not enable an exchange of non-repudiation evidences.
The protocol always ends up aborted as $a_1^\O$ is
the only move that satisfies $\varphi_\O$ for $\O$
in state $\set{M_1, \eoo}$ against all behaviors of $\R$ and the $\ttp$; 
once $\O_1$ sends her signature in $m_3$, there is no move available 
to $\O_1$ such that satisfaction of $\varphi_\R \wedge \varphi_\ttp$ 
is guaranteed to satisfy $\varphi_\O$, as $\R$ may decide to stop 
participating in the protocol.

\medskip\noindent{\bf Analysis of the refinement $\textrm{P}_2$.}
In this case, $\R$ has no ability to resolve the protocol.
It is easy to verify that $\pmin \preceq P_2 \preceq \pmax$.
Therefore, $P_2 \in \pb$ and hence, by Lemma~\ref{lem-alt},
$P_2 \in P_{AGS}$.
This protocol is a fair non-repudiation protocol that satisfies
fairness, balance and timeliness.
If $\O$ does not send $m_3$, then $\R_2$ has no choice of moves.
But since $P_2$ satisfies bounded idle time, $\O_2$ will
eventually either abort or resolve the protocol.
As $\ttp_2$ satisfies the AGS constraints on the $\ttp$, 
either both agents get abort tokens or they get their respective 
non-repudiation evidences eventually.

\medskip\noindent{\bf Analysis of the refinement $\textrm{P}_3$.}
Since $\O$ has no ability to abort the protocol, while both agents
have the ability to resolve it, the predicates $\abo$ and $\abr$
are always false.
The agent and $\ttp$ objectives then reduce to,
\begin{gather*}
\varphi_\O = \diam M_1 \wedge \bo(\diam \sigb^\R \vee \diam \sigb^\ttp), \\
\varphi_\R = \bo (\eoo \im (\diam \siga^\O \vee \diam \siga^\ttp)), \\
\varphi_\ttp = \bo (\reqr \im (\diam \siga^\ttp \vee \diam \sigb^\ttp)) 
\wedge \bo (\siga^\ttp \im \diam \sigb^\ttp) \wedge \\
\bo (\sigb^\ttp \im \diam \siga^\ttp) \eqpun .
\end{gather*}
\begin{table*}[!t]
\begin{center}
\begin{tabular}{|l||l|l||l|l|l|}
\hline
\bf{Delivered message sequences \quad} & \multicolumn{5}{c|}{\bf{Moves for $\O_3$ and $\R_3$}} \\
\cline{2-6}
& \multicolumn{2}{c||}{\bf{Choices for $\O_3$}} & \multicolumn{3}{c|}{\bf{Choices for $\R_3$}} \\
\hline \hline
$\bf{\seq{}}$ & $m_1 \qquad$ & $m_1 \qquad$ & $\idle \qquad$ & $\idle \qquad$ & $\idle \qquad$ \\
$\bf{\seq{m_1}}$ & $\idle$ & $\idle$ & $m_2$ & $r_1^\R$ & \text{either} $m_2$ \text{or} $r_1^\R$ \\
$\bf{\seq{m_1, m_2}}$ & $m_3$ & $r_1^\O$ & $\idle$ & $r_1^\R$ & \text{either} $\idle$ \text{or} $r_1^\R$ \\
$\bf{\seq{m_1, m_2, m_3}}$ & $\idle$ & $r_1^\O$ & $\idle$ & $m_4$ & $r_1^\R$ \\
\hline
\end{tabular}
\end{center}
\caption[]{The moves that satisfy the objectives of assume-guarantee 
synthesis for $\O_3$ and $\R_3$ are shown in this table at relevant
protocol states represented by message sequences, when the agents have 
no ability to abort the protocol.}
\label{table-stra-noa}
\end{table*}  
The moves of the agents that satisfy the objectives of
assume-guarantee synthesis at select protocol valuations represented
by message sequences are shown in Table~\ref{table-stra-noa}.
It is easy to verify that as $\pmin \not\preceq P_3 \preceq \pmax$,
$P_3 \not\in \pb$ and hence by Lemma~\ref{lem-alt}, 
$P_3 \not\in P_{AGS}$.
Since $\ttp_3$ satisfies the AGS constraints on the $\ttp$,
$P_3$ is a fair non-repudiation protocol similar to the
ZG optimistic non-repudiation protocol, but it does not satisfy 
timeliness \cite{KremerMZ02} as $\O$ does have the ability to abort 
the protocol.
If message $m_1$ is not delivered, then $\O$ has no choice of moves
to satisfy $\varphi_\O$, while $\varphi_\R \wedge \varphi_\ttp$
are satisfied trivially.
Balance does not apply in this case as there are no abort moves.

\medskip\noindent{\bf Analysis of the refinement $\pmax$.}
In the maximal refinement $\pmax = (\omax, \rmax, \tmax)$,
since $\tmax$ satisfies the AGS constraints on the $\ttp$, if her
first response to an abort or resolve request is $\comp{x, y}$,  
she can choose any move in $\set{\idle, x, y, \comp{x, y}}$ for
all subsequent abort or resolve requests.
Consider a refinement $P_{KM} = (\O_{KM}, \R_{KM}, \ttp_{KM}) \preceq \pmax$, 
where $\O_{KM}$ and $\R_{KM}$ correspond to $\omax$ and $\rmax$ and
$\ttp_{KM} \preceq \tmax$ such that $\ttp_{KM}$ goes idle after 
her first response to an abort or resolve request.
$P_{KM}$ is then the KM protocol.
We remark that given the choices of moves for the $\ttp$ after her 
first response as suggested by assume-guarantee synthesis, choosing 
$\idle$ satisfies the informal notion of efficiency.
This refinement ensures fairness, balance and timeliness.

%% file: symmetric.tex
\section{A Symmetric Fair Non-Repudiation Protocol}

In the KM, ASW and GJM protocols, $\R$ cannot abort the protocol.
While the ability of $\O$ to abort the protocol after sending $m_1$
is required in the event $m_1$ is not delivered or $\R$ does not
send $m_2$, it can be used to abort the protocol even if all channels are 
resilient or if $\O$ decides not to sign the contract after receiving $m_2$.
The protocols give $\O$ the ability to postpone abort decisions but
deny $\R$ a similar ability.
While this does not violate fairness or abuse-freeness as per
prevailing definitions, it is not equitable to both agents.
If $\R$ does not want to participate in a protocol instance, 
then the only choice of moves for $\R$ is $\idle$ and not $m_2$;
$\O$ will then eventually abort the protocol.
Once $m_2$ has been sent, if $\R$ decides not to participate in
the protocol and not be held responsible for signing the contract, he 
has no choice of moves.
If he decides to ignore $m_3$, then $\O$ will resolve the protocol
resulting in non-repudiation evidences being issued to $\O$,
using which she can claim $\R$ is obligated by the contract.

In this section we present a symmetric fair non-repudiation 
protocol that gives $\R$ the ability to abort the protocol, assuming 
that the channels between the agents and the $\ttp$ are operational.
If we enhance the ability of $\R$ by including an
abort move $a_1^\R$ without enhancing $\O$ and the $\ttp$, 
then assume-guarantee synthesis fails.
By enhancing both $\O$ and the $\ttp$, using assume-guarantee
analysis, we design a new fair non-repudiation protocol that
(a) has no $Y$-attack for all $Y \subseteq\set{\O,\R}$;
and (b) that provides $\R$ the ability to abort. 
In the following, we show that if we fix the behavior of the $\ttp$, 
ensuring $\ttp$ inviolability, then the protocol is attack-free.

\begin{algorithm}[t]
  \caption{$\textsc{Main Protocol of our Symmetric Non-repudiation Protocol}$}
  \label{prot:symm-main}
    $\O$ sends $m_1$ to $\R$\;
    \eIf{($\R$ does not want to participate)}{
      $\R$ sends $a_1^\R$ to the $\ttp$\;
    }{
      $\R$ sends $m_2$ to $\O$\;
      \eIf{($\R$ does not send $m_2$ on time)}{
        $\O$ sends $a_1^\O$ to the $\ttp$\;
      }{
        $\O$ sends $m_3$ to $\R$\;
        \eIf{($\O$ does not send $m_3$ on time)}{
          \eIf{($\R$ does not want to participate)}{
            $\R$ sends $a_1^\R$ to the $\ttp$\;
          }{
            $\R$ sends $r_1^\R$ to the $\ttp$\;
          }
        }{
          $\R$ sends $m_4$ to $\O$\;
          \If{($\R$ does not send $m_4$ on time)}{
            $\O$ sends $r_1^\O$ to the $\ttp$\;
          }
        }
      }
    }
\end{algorithm}

\begin{algorithm}[t]
  \caption{$\textsc{Abort Subprotocol}.\ X \in \set{\O, \R}$}
  \label{prot:symm-abort}
    $X$ sends $a_1^X$ to $\ttp$\;
    \eIf{(the protocol has been aborted or resolved)}{
      $\ttp$ goes idle\;
    }{
      \eIf{($X = \R$)}{
        $\ttp$ sends $req^\O$ to $\O$\;
        \eIf{($\O$ sends $res^\O$ on time)}{
          $\ttp$ marks this protocol instance as resolved in its persistent DB\;
          $\ttp$ sends $\comp{r_2^\O, r_2^\R}$ to $\O$ and $\R$\;
        }{
          $\ttp$ marks this protocol instance as aborted in its persistent DB\;
          $\ttp$ sends $\comp{a_2^\O, a_2^\R}$ to $\O$ and $\R$\;
        }
      }{
        $\ttp$ marks this protocol instance as aborted in its persistent DB\;
        $\ttp$ sends $\comp{a_2^\O, a_2^\R}$ to $\O$ and $\R$\;
      }
    }
\end{algorithm}

Consider the following refinement $\psm = (\osm, \rsm, \tsm)$
with $\pmax \preceq \psm$ defined as follows:
\begin{itemize}[]
\item $\moves_{\osm} = \moves_{\omax} \cup \set{res^\O}$;
\item $\moves_{\rsm} = \moves_{\rmax} \cup \set{a_1^\R}$; and
\item $\moves_{\tsm} = \moves_{\tmax} \cup \set{req^\O}$.
\end{itemize}
The move $req^\O$ may be sent by $\tsm$ only after receiving
an abort request from $\R$.
The move $res^\O$ may be sent by $\osm$ only after receiving
$req^\O$.
We present the main protocol and the abort subprotocol for our
symmetric fair non-repudiation protocol in Protocol~\ref{prot:symm-main}
and Protocol~\ref{prot:symm-abort}; the resolve 
subprotocol is identical to the one in the KM protocol.

To facilitate the assume-guarantee analysis of $\psm$, we present the 
following {\em enhanced AGS constraints on the $\ttp$\/} that is both
necessary and sufficient to ensure $\ttp$ inviolability (neither 
agent can violate $\varphi_\ttp$):
\begin{enumerate}
\item {\em Abort constraint.\/}
If the first request received by the $\ttp$ is $a_1^\O$, then her 
response to that request should be $\comp{a_2^\O, a_2^\R}$;
If the first request received by the $\ttp$ is $a_1^\R$, then her
response to that request should be $req^\O$;
\item {\em Resolve constraint.\/}
If the first request received by the $\ttp$ is a resolve
request, then her response to that request should be 
$\comp{r_2^\O, r_2^\R}$;
If the $\ttp$ receives $res^\O$ in response to $req^\O$ within
bounded idle time, then her response should be $\comp{r_2^\O, r_2^\R}$,
otherwise it should be $\comp{a_2^\O, a_2^\R}$.
\item {\em Accountability constraint.\/}
If the first response from the $\ttp$ is 
$\comp{x, y}$ or the first response from the $\ttp$ is $req^\O$ and
the next response is $\comp{x, y}$, then for all subsequent abort or 
resolve requests her response should be in the set 
$\set{\idle, x, y, \comp{x, y}}$.
\end{enumerate}
The enhanced AGS constraints on the $\ttp$ are required both to
satisfy the implication condition $(\varphi_\O \wedge \varphi_\R)
\im \varphi_\ttp$ and the condition for weak co-synthesis,
$(\varphi_\O \wedge \varphi_\R \wedge \varphi_\ttp)$.
Since $\tsm$ waits for a bounded number of turns before sending
abort tokens to both agents after sending $req^\O$, we require that
(a) the channels between the agents and the $\ttp$ are operational, and 
(b) the time taken to deliver messages $req^\O$ and
$res^\O$ be subsumed by the bound on idle time chosen by the $\ttp$ 
between sending $req^\O$ and abort tokens.
As there is no bound on the time taken to deliver messages on
resilient channels, the above AGS constraints on the $\ttp$ cannot 
be enforced without operational channels.
Consider a partial trace that ends in protocol state $\set{M_1, \eoo,
M_2, \eor, M_3}$; messages $m_1$ and $m_2$ have been received and
$m_3$ has been sent.
If $\R$ now aborts the protocol and the $\ttp$ sends $req^\O$ to $\O$,
then resilient channels can delay delivering either $req^\O$ or $res^\O$ 
sufficiently for the $\ttp$ to abort the protocol.
In this case if $m_3$ is eventually delivered, $\varphi_\O$ is violated
whereas $\varphi_\R \wedge \varphi_\ttp$ is satisfied.

In the following lemma we show that in $\psm$, $\O$ cannot violate
$\varphi_\R$ while satisfying $\varphi_\O$, $\R$ cannot violate
$\varphi_\O$ while satisfying $\varphi_\R$ and $\O$ and $\R$ cannot
violate $\varphi_\ttp$ while satisfying their objectives.
That is, in the refinement $\psm$ we have 
$\semb{\O \para \R \para \tsm \para \Sc}
\subseteq (\varphi_\O \wedge \varphi_\R) \im \varphi_\ttp$, and
$\semb{\O \para \rsm \para \ttp \para \Sc}
\subseteq (\varphi_\O \wedge \varphi_\ttp) \im \varphi_\R$.
However, it is not the case that
$\semb{\osm \para \R \para \ttp \para \Sc}
\subseteq (\varphi_\R \wedge \varphi_\ttp) \im \varphi_\O$.
But if the $\ttp$ is fixed then the implication condition holds, 
i.e., $\semb{\osm \para \R \para \tsm \para \Sc}
\subseteq \varphi_\R \im \varphi_\O \subseteq 
(\varphi_\R \wedge \varphi_\ttp) \im \varphi_\O$.
It follows that under the assumption that the $\ttp$ does not change 
her behavior, while satisfying her objective, the symmetric protocol 
is attack-free.
We present the following lemma and prove it in the Appendix.

\begin{lem}{}
\label{psm-is-ags}
For the refinement $\psm = (\osm, \rsm, \tsm)$, if the channels
between the agents and the $\ttp$ are operational, then there
exists no $Y$-attack for all $Y \subseteq \set{\O,\R}$.
\end{lem}

The assumption that the bound on idle time of the $\ttp$ between 
sending $req^\O$ and abort tokens subsume the time taken for the 
delivery of messages $req^\O$ and $res^\O$ can easily be enforced 
before the beginning of a protocol; $\O$ agrees to participate in 
the protocol with a given $\ttp$, only if the bound chosen by the 
$\ttp$ is satisfactory.
We point out that in state $\set{\eoo, M_2}$, if $\R$ sends an 
abort request, he still needs $\O$'s co-operation to abort the
protocol.
Since she has $m_2$, she can launch recovery if she so desires 
by composing $res^\O$ when she receives $req^\O$.
But this is identical to the ability of $\O$ in aborting the protocol
after she sends $m_1$.
$\R$ can resolve the protocol as soon as he receives $m_1$ and thus 
hold $\O$ as a signatory to the contract even if she decided to abort
the protocol after sending $m_1$.
The protocol is therefore symmetrical to both $\O$ and $\R$.
In addition, we claim that this version of the protocol provides
better quality of service in terms of timeliness; $\O$ does not
have to wait after sending $m_1$ for $\R$ to send $m_2$,
in protocol instances where $\R$ has no desire to sign the 
contract.
The following theorem states that if the $\ttp$ does not change
her behavior, then the refinement $\psm$ is an attack-free fair 
non-repudiation protocol.
The proof is in the Appendix.

\begin{theo}{\ (Symmetric attack-free protocol)} 
\label{theo-psm-nonrep}
Given the channels between the agents and the $\ttp$ are operational
and the $\ttp$ does not deviate from satisfying the enhanced AGS 
constraints on the $\ttp$, the refinement $\psm = (\osm, \rsm, \tsm)$ 
is an attack-free fair non-repudiation protocol.
\end{theo}

\medskip\noindent{\bf From $P_{AGS}$ to $\psm$.}
We can systematically analyze refinements leading to $\psm$.
Similar to the case of synthesizing the KM non-repudiation protocol,
we now present the steps that explore refinements leading to
$\psm$.
We assume the $\ttp$ satisfies the AGS constraints on the $\ttp$ and all
refinements satisfy bounded idle time.
The analyzed refinements are as follows:
\begin{enumerate}
\item $\pmin = (\omin, \rmin, \tmin)$; the minimal refinement.
\item $P_1 = (\O_1, \R_1, \ttp_1)$ with \\
$\moves_{\O_1} = \moves_{\omin} \cup \set{\idle, m_3}$,
$\moves_{\R_1} = \moves_{\rmin} \cup \set{m_2, m_4}$ and 
$\ttp_1 = \tmax$.
\item $P_2 = (\O_2, \R_2, \ttp_2)$ with \\
$\moves_{\O_2} = \moves_{\O_1} \cup \set{r_1^\O}$,
$\moves_{\R_2} = \moves_{\R_1}$ and
$\ttp_2 = \tmax$.
\item $P_3 = (\O_3, \R_3, \ttp_3)$ with \\
$\moves_{\O_3} = \moves_{\O_2} \setm \set{a_1^\O}$,
$\moves_{\R_3} = \moves_{\R_1} \cup \set{r_1^\R}$ and
$\ttp_3 = \tmax$.
\item $\pmax = (\omax, \rmax, \tmax)$; the maximal refinement.
\item $\psm = (\osm, \rsm, \tsm)$ with \\
$\moves_{\osm} = \moves_{\omax} \cup \set{res^\O}$,
$\moves_{\rsm} = \moves_{\rmax} \cup \set{a_1^\R}$ and \\
$\moves_{\tsm} = \moves_{\tmax} \cup \set{req^\O}$.
\end{enumerate}

\smallskip\noindent{\bf Implementation.}
We have implemented a prototype for assume-guarantee synthesis of
fair non-repudiation protocols.
Our implementation considers triples of refinements $O' \preceq \O$, 
$R' \preceq \R$, and $\ttp' \preceq \ttp$ and then explores all possible 
message sequences given these participant refinements.
We implemented a scheduler that backtracks and systematically
schedules all participants at all protocol states.
Using the scheduler, given a subset of participant refinements, with all 
other participants being most general, the implementation explores all 
possible traces and checks if each trace satisfies the required 
AGS conditions.
Note that in checking the satisfaction of the AGS conditions, for the 
implication conditions we need to consider the most general participants 
against each of the refinements $\O'$, $\R'$ and $\ttp'$.
The checking of the implication conditions is achieved by solving 
secure equilibrium on graph games with lexicographic objectives.
Our implementation generates all possible AGS solutions.
The analysis of the AGS solutions generated by our implementation was key
in obtaining the symmetric protocol; using a procedure similar to 
obtaining $P_{KM}$ from $P_{AGS}$.

%% file: conclusion.tex
\section{Conclusion}

In this work we introduce and demonstrate the effectiveness of 
assume-guarantee synthesis in synthesizing fair exchange protocols. 
Our main goal is to introduce a general 
assume-guarantee synthesis framework that can be used with a
variety of objectives; we considered a 
$\ttp$ objective that treats the agents symmetrically, but the 
framework can be used with possibly weaker $\ttp$ objectives
that treat agents asymmetrically. 
Using assume-guarantee analysis we have obtained a new symmetric 
protocol that is attack-free, given the channels to the $\ttp$ 
are operational. 
While the need for operational channels may be considered 
impractical, we remark that it is this flexible framework  
that could automatically generate such protocols of theoretical 
interest in the first place.
For future work we will study the application of assume-guarantee 
synthesis to other security protocols.

%% file: appendix.tex
\newpage

\section{Appendix}

\medskip\noindent{\bf Translating protocol models to process models.}
We now present a translation from the protocol model introduced in
Section~\ref{sec:fair-nonrep} to the process model introduced in
Section~\ref{sec:co-synthesis}.
We take $\moves = \calm$, as the set of process moves, corresponding
to the set of all messages in $\calm$.
For $1 \le i \le n$, we map each participant $A_{i - 1}$ to a process 
$P_i$ as follows:
\begin{itemize}
\item $X_i = V_{i - 1} \cup \set{L_i}$, is the set of variables of process
$P_i$ that includes all participant variables $V_{i - 1}$ and a 
special variable $L_i$ corresponding to control points, taking finitely
many values in $\nats$,
\item for all valuations $f \in \Valua_i[X_i]$, we have
$\mov_i(f) = \mas_{i - 1}(f \obciach V_{i - 1})$ and
\item $\trans_i: \Valua_i[\set{L_i}] \times 
\Valua_i[X_i \setm \set{L_i}] \times \moves \mapsto 
\Valua_i[\set{L_i}] \times \Valua_i[X_i \setm \set{L_i}]$
is the process transition function that exactly corresponds to 
the participant transition function $\Lambda_{i - 1}$.
\end{itemize}
The sets $X_i$ form a partition of $X = \bigcup_{i = 1}^n X_i$.
The set of processes $P_i$, given all possible behaviors
of a fair scheduler $\Sc$, corresponds to the most general 
exchange program.
The realization of a protocol corresponds to a refinement 
$P_i' \preceq P_i$ for $1 \le i \le n$, where each participant
$A_{i - 1}'$ maps to the process $P_i'$ as follows:
\begin{itemize}
\item $X_i' = X_i = V_{i - 1} \cup \set{L_i}$ is the set of variables of
process $P_i'$,
\item for all valuations $f \in \Valua_i'[X_i']$, we have
$\mov_i'(f) = \mas_{i - 1}'(f \obciach V_{i - 1})$ and 
\item for all valuations $f \in \Valua_i'[X_i']$, 
for all moves $m \in \moves$, we have
$\trans_i'(f, m) = \Lambda_{i - 1}'(f(L_i), f \obciach V_{i - 1}, m)$.
\end{itemize}
A protocol instance (protocol run) is a trace in 
$\semb{P_1' \para P_2' \ldots \para P_n' \para \Sc}(v_0)$ for an initial 
valuation $v_0 \in \Valua[X]$.
The specifications of the participants, which were defined as
a set of desired sequences of messages, are subsets of traces
in $\semb{P_1' \para P_2' \ldots \para P_n' \para \Sc}(v_0)$.
Given specifications $\varphi_i$ for process $P_i$, a
$Y$-attack for $Y \subseteq \set{P_1, P_2, \ldots, P_n}$ satisfies 
$\varphi_i$ for all $P_i \in Y$, while violating $\varphi_j$
for at least one process 
$P_j' \in (\set{P_1, P_2, \ldots, P_n} \setm Y)'$.
There are three participants in a two party fair non-repudiation
protocol, the originator $\O$, the recipient $\R$ and the
trusted third party $\ttp$.
We therefore take $n = 3$ in modeling two party fair exchange
protocols in the above translation.

We now prove Lemma~\ref{lem-alt}.
Given a refinement $P' = (\O', \R', \ttp') \preceq P$, we first
characterize the smallest restriction on $\O'$ and $\R'$ that
satisfy the implication conditions:
\begin{align}
&\semb{\O \para \R' \para \ttp \para \Sc} \subseteq 
(\varphi_\O \wedge \varphi_\ttp) \im \varphi_\R; \text{ and} 
\label{r-implication} \\
&\semb{\O' \para \R \para \ttp \para \Sc} \subseteq 
(\varphi_\R \wedge \varphi_\ttp) \im \varphi_\O \eqpun . 
\label{o-implication}
\end{align}
We show that for all refinements $\R' \preceq \R$, the implication
condition (\ref{r-implication}) holds.
In order to characterize the smallest restrictions on $\O$ 
that satisfies the implication condition (\ref{o-implication}), we 
recall the following constraints on $\O$.
We show that these constraints are both necessary and sufficient
to satisfy (\ref{o-implication}).

\medskip\noindent{\bf AGS constraints on $\O$.}
We say that a refinement $\O' \preceq \O$ satisfies
the {\em AGS constraints on $\O$\/} if $\O'$ satisfies
the following constraints:
\begin{enumerate}
\item $a_1^\O \not\in \mov_{\O'}(v_0)$;
\item $\siga^\O \not\in \mov_{\O'}(\set{M_1, \eor, \reqa^\O})$; and
\item $a_1^\O \not\in \mov_{\O'}(\set{M_1, \eor, M_3})$.
\end{enumerate}

\medskip\noindent{\bf The most flexible refinements $\O' \preceq \O$
and $\R' \preceq \R$.}
We now characterize the most flexible refinements $\O' \preceq \O$
and $\R' \preceq \R$ that satisfy the implication conditions 
$(\varphi_\R \wedge \varphi_\ttp) \im \varphi_\O$ and
$(\varphi_\O \wedge \varphi_\ttp) \im \varphi_\R$.

\begin{lem}{} 
\label{lem-max-r}
For all refinements $\R' \preceq \R$, the following assertion holds:
\[
\semb{\O \para \R' \para \ttp \para \Sc} \subseteq 
(\varphi_\O \wedge \varphi_\ttp) \im \varphi_\R.
\]
\end{lem}
\begin{proof}
Consider an arbitrary refinement $\R' \preceq \R$.
We have the following cases of sets of traces of
$\semb{\O \para \R' \para \ttp \para \Sc}$ for the proof:
\begin{itemize}
\item {\em Case 1. Set of traces where $m_3$ has been received.\/}
For all traces where $m_3$ has been received, $\varphi_\R$ is satisfied. 
Therefore all these traces satisfy the implication condition,
$(\varphi_\O \wedge \varphi_\ttp) \im \varphi_\R$.
\item {\em Case 2. Set of traces where $m_3$ has not been received.\/}
For all traces where $m_3$ has not been received, the traces where 
either $\varphi_\O$ or $\varphi_\ttp$ is violated, satisfy
the implication condition
$(\varphi_\O \wedge \varphi_\ttp) \im \varphi_\R$ trivially.
The interesting case are those traces that satisfy 
$\varphi_\O \wedge \varphi_\ttp$ but violate $\varphi_\R$.
These are exactly the traces where $\O$ does not have 
$\sigb^\R$, since $m_4$ is not sent before receiving $m_3$,
and $\R$ does not have $\siga^\O$, as otherwise $\varphi_\R$
would be satisfied.
We have following cases that lead to a contradiction:
\begin{itemize}
\item {\em Case (a). $\O$ aborts the protocol.\/}
In these traces, since $\varphi_\ttp$ is satisfied, the abort token 
must have been sent to both agents, and since neither agent will be
sent the other's signature and the channels between the agents and
the $\ttp$ are resilient, the traces satisfy $\varphi_\R$, leading
to a contradiction.
\item {\em Case (b). $\O$ or $\R'$ resolve the protocol.\/}
In these traces, since $\varphi_\ttp$ is true, the $\ttp$ sends 
$\siga^\ttp$ to $\R$ and $\sigb^\ttp$ to $\O$ and never sends either 
$\abo$ or $\abr$.
This implies, given the channel between the agents and the $\ttp$
is resilient, the traces satisfy $\varphi_\R$, leading to a
contradiction.
\item {\em Case (c). $\R'$ chooses move $\idle$.\/}
In these traces, since $\varphi_\O$ is true, either $\O$ aborts
the protocol after sending $m_1$ or she chooses to abort or resolve
the protocol after receiving $m_2$. 
In either case, given the traces satisfy $\varphi_\ttp$, by the above 
argument $\varphi_\R$ is satisfied as well, irrespective of the
behavior of the channel between $\O$ and $\R$.
This again leads to a contradiction.
\end{itemize}
\end{itemize}
Since we have shown that for all traces, either $\varphi_\R$ is
satisfied or satisfaction of $\varphi_\O \wedge \varphi_\ttp$ 
implies satisfaction of $\varphi_\R$, we conclude that
for all refinements $R' \preceq R$ the assertion holds.
\qed
\end{proof}

It follows from Lemma~\ref{lem-max-r}, that as $\R'$ can always 
resolve the protocol in state $\set{\eoo}$ and all successor states, 
such that the resulting trace satisfies
$(\varphi_\O \wedge \varphi_\ttp) \im \varphi_\R$, we have
$m_2 \in \mov_{\R'}(\set{\eoo})$.
Similarly, $m_4 \in \mov_{\R'}(\set{\eoo, M_2, \siga^\O})$ as
$\varphi_\R$ is satisfied in all traces where $m_3$ has been received,
thus satisfying $(\varphi_\O \wedge \varphi_\ttp) \im \varphi_\R$.

In the following lemma, in assertion 1 we show that for all
refinements $\O' \preceq \O$ that satisfy the AGS constraints on $\O$,
the implication condition (\ref{o-implication}) is satisfied; in 
assertion 2 we show that if $\O'$ does not satisfy the AGS constraints 
on $\O$, the implication condition (\ref{o-implication}) is violated.

\begin{lem}{\ (The smallest restriction on $\O' \preceq \O$)} 
\label{lem-max-or}
For all refinements $\O' \preceq \O$, the following assertions hold:
\begin{enumerate}
\item if $\O'$ satisfies the AGS constraints on $\O$, then
\[
\semb{\O' \para \R \para \ttp \para \Sc} \subseteq 
(\varphi_\R \wedge \varphi_\ttp) \im \varphi_\O.
\]
\item if $\O'$ does not satisfy the AGS constraints on $\O$, then
\[
\semb{\O' \para \R \para \ttp \para \Sc} \not\subseteq
(\varphi_\R \wedge \varphi_\ttp) \im \varphi_\O.
\]
\end{enumerate}
\end{lem}
\begin{proof}
Consider an arbitrary refinement $\O' \preceq \O$ that satisfies
the AGS constraints on $\O$.
We have the following cases of sets of traces of
$\semb{\O' \para \R \para \ttp \para \Sc}$ for the proof:
\begin{itemize}
\item {\em Case 1. Set of traces where $m_4$ has been received.\/}
In the case of classical co-synthesis, an adversarial $\R$ will never 
send $m_4$ as that satisfies $\varphi_\O$ unconditionally, but in 
assume-guarantee synthesis, from Lemma~\ref{lem-max-r}, 
since all refinements of $\R$ satisfy the weaker condition of
$(\varphi_\O \wedge \varphi_\ttp) \im \varphi_\R$,
$m_4 \in \mov_{\R'}(\seq{\eoo, M_2, \siga^\O})$.
For all traces where $m_4$ has been received, $\varphi_\O$ is satisfied. 
Therefore all these traces satisfy the implication condition
$(\varphi_\R \wedge \varphi_\ttp) \im \varphi_\O$.
\item {\em Case 2. Set of traces where $m_4$ has not been received.\/}
For all traces where $m_4$ has not been received, the traces where 
either $\varphi_\R$ or $\varphi_\ttp$ is violated, satisfy
the implication condition
$(\varphi_\R \wedge \varphi_\ttp) \im \varphi_\O$ trivially.
The interesting case are those traces that satisfy 
$\varphi_\R \wedge \varphi_\ttp$ but violate $\varphi_\O$.
These are exactly the traces where $\O$ does not have 
$\sigb^\R$, since $m_4$ has not been received.
We have the following cases that lead to a contradiction:
\begin{itemize}
\item {\em Case (a). $\O'$ has sent $m_3$.\/} 
In these traces, since $\O'$ satisfies the AGS constraints on $\O$,
the only choice of moves for $\O'$ are $\idle$ or $r_1^\O$; she can 
wait for $\R$ to send $m_4$ or resolve the protocol.
In the set of traces where she eventually receives $m_4$, by Case 1, 
the traces satisfy
$(\varphi_\R \wedge \varphi_\ttp) \im \varphi_\O$.
If she does not receive $m_4$, she will eventually resolve the
protocol to satisfy $\varphi_\O$.
In the set of traces where she eventually resolves the protocol,
since $\varphi_\ttp$ is satisfied, and $\R$ cannot abort the
protocol, the $\ttp$ will eventually respond to her request by sending 
her non-repudiation evidence and not the abort token.
These traces therefore satisfy $\varphi_\O$, leading to a contradiction.
\item {\em Case (b). $\O'$ aborts the protocol before sending $m_3$.\/}
Since $\O'$ satisfies the AGS constraints on $\O$, she cannot abort
the protocol in the initial state $v_0$.
Therefore, $\O'$ must have started the protocol by sending $m_1$.
In all these traces, $\O'$ aborts the protocol after sending $m_1$ but
before sending $m_3$ and since $\O'$ satisfies the AGS constraints on 
$\O$, she will not send $m_3$ after sending the abort request.
Since these traces satisfy $\varphi_\ttp$, the abort token 
must have been sent to both agents, and since neither agent will be
sent the other's signature and the channels between the agents and
the $\ttp$ are resilient, the traces satisfy $\varphi_\O$, leading
to a contradiction.
\item {\em Case (c). $\O'$ resolves the protocol before sending $m_3$.\/}
In these traces, since $\varphi_\ttp$ is true, the $\ttp$ sends 
$\sigb^\ttp$ to $\O$ and $\siga^\ttp$ to $\R$ and never sends either 
$\abo$ or $\abr$.
This implies, given the channel between the agents and the $\ttp$
is resilient, the traces satisfy $\varphi_\O$, leading to a
contradiction.
\item {\em Case (d). $\O'$ chooses move $\idle$ instead of sending $m_3$.\/}
In these traces, since $\varphi_\R$ is true, $\R$ must have resolved
the protocol after receiving $m_1$. 
In this case, given the traces satisfy $\varphi_\ttp$, by the above 
argument $\varphi_\O$ is satisfied as well.
This again leads to a contradiction.
\item {\em Case (e). The channel between $\O$ and $\R$ is unreliable.\/}
If either $m_1$ or $m_2$ are not delivered, then $\O'$ can abort the
protocol.
If either $m_3$ or $m_4$ are not delivered, then $\O'$ can resolve the
protocol.
In either case, by Case (a), Case (b) and Case (c), we have
$\varphi_\O$ is satisfied even when the channel between $\O$ and $\R$
is unreliable, leading to a contradiction.
\end{itemize}
\end{itemize}
We conclude that for all $\O'$ that satisfy the AGS constraints on
$\O$, we have $\semb{\O' \para \R \para \ttp \para \Sc} \subseteq
(\varphi_\R \wedge \varphi_\ttp) \im \varphi_\O$.

For assertion 2, consider an arbitrary refinement $\O' \preceq \O$
that does not satisfy the AGS constraints on $\O$.
We consider violation of the constraints on a case by case basis.
For each case we produce a witness trace that violates the
implication condition $(\varphi_\R \wedge \varphi_\ttp) \im \varphi_\O$.
We proceed as follows:
\begin{itemize}
\item {\em Case 1. $a_1^\O \in \mov_{\O'}(v_0)$.\/}
In a trace where $\O'$ sends an abort request before sending message 
$m_1$ in the initial protocol state $v_0$, it is trivially the case
that the trace does not satisfy $\varphi_\O$ but satisfies 
$\varphi_\R$.
If the $\ttp$ satisfies the AGS constraints on the $\ttp$ and sends
$\comp{a_2^\O, a_2^\R}$ in response, then the trace satisfies
$\varphi_\ttp$.
Therefore, the trace violates $(\varphi_\R \wedge \varphi_\ttp) \im
\varphi_\O$.
\item {\em Case 2. $\siga^\O \in \mov_{\O'}(M_1, \eoo, \reqa^\O)$.\/}
To produce a witness trace we consider a partial trace that ends
in protocol state $\set{M_1, \eoo, \reqa^\O}$; messages $m_1$ and $m_2$
have been received and $a_1^\O$ has been sent.
Since the channel between $\O$ and the $\ttp$ is resilient, the
abort request is eventually processed by the $\ttp$.
If $\O'$ sends message $m_3$ in this state and the $\ttp$ responds with 
move $\comp{a_2^\O, a_2^\R}$ to her abort request, then
there exists a behavior of the channel between $\O$ and $\R$ such
that $m_3$ is eventually delivered and the protocol is aborted.
The trace therefore satisfies $\varphi_R \wedge \varphi_\ttp$ but
violates $\varphi_\O$; as $\O$ cannot get $\R$'s signature after the
protocol is aborted and $\R$ has her signature.
\item {\em Case 3. $a_1^\O \in \mov_{\O'}(M_1, \eoo, M_3)$.\/}
To produce a witness trace we consider a partial trace that ends in
protocol state $\set{M_1, \eoo, M_3}$; messages $m_1$ and $m_2$ have
been received and $m_3$ has been sent.
If $\O'$ aborts the protocol in this state and the $\ttp$ satisfies
the AGS constraints on the $\ttp$ and responds with 
move $\comp{a_2^\O, a_2^\R}$, then there exists a behavior of the 
channel between $\O$ and $\R$, where $m_3$ is eventually delivered to 
$\R$.
The trace satisfies $\varphi_\R \wedge \varphi_\ttp$ but violates
$\varphi_\O$.
\end{itemize}
We conclude that if $\O'$ does not satisfy the AGS constraints on
$\O$, then 
$\semb{\O' \para \R \para \ttp \para \Sc} \not\subseteq 
(\varphi_\R \wedge \varphi_\ttp) \im \varphi_\O$.
\qed
\end{proof}

From Lemma~\ref{lem-max-or}, it is both necessary and sufficient
that $\O$ satisfies the AGS constraints on $\O$ to ensure
the implication condition (\ref{o-implication}).

\medskip\noindent{\bf The maximal refinement $\pmax = (\omax, \rmax, \tmax)$.}
We recall the definition of the maximal refinement 
$\pmax = (\omax, \rmax, \tmax)$ below:
\begin{enumerate}
\item $\omax \preceq \O$ satisfies the AGS constraints on $\O$
and for all $\O'$ that satisfy the constraints, we have
$O' \preceq \omax$;
\item $\rmax = \R$; and
\item $\tmax \preceq \ttp$ satisfies the AGS constraints on
the $\ttp$ and for all $\ttp'$ that satisfy the constraints,
we have $\ttp' \preceq \tmax$.
\end{enumerate}

\medskip\noindent{\bf The weak co-synthesis requirement.}
Let $b \in \nats$ be a bound on the number of times that $\O$
or the $\ttp$ may choose the idle move $\idle$ when
scheduled by $\Sc$.
In the following lemma, for all refinements $P' \preceq \pmax$ that 
satisfy the AGS constraints on the $\ttp$, in assertion 1 we show that 
if $b$ is finite, then the condition for weak co-synthesis is satisfied;
in assertion 2 we show that if $b$ is unbounded, then the condition
for weak co-synthesis is violated.

\begin{lem}{\ (Bounded idle time lemma)} \label{lem-bounded-idle}
For all refinements $P' = (\O', \R', \ttp') \preceq \pmax$ that
satisfy the AGS constraints on the $\ttp$, for all $b \in \nats$ 
with $\O'$ and $\ttp'$ choosing
at most $b$ idle moves when scheduled by $\Sc$, the following
assertions hold:
\begin{enumerate}
\item if $b$ is finite, then
$\semb{\O' \para \R' \para \ttp' \para \Sc} \subseteq 
(\varphi_\O \wedge \varphi_\R \wedge \varphi_\ttp)$.
\item if $b$ is unbounded, then
$\semb{\O' \para \R' \para \ttp' \para \Sc} \not\subseteq 
(\varphi_\O \wedge \varphi_\R \wedge \varphi_\ttp)$.
\end{enumerate}
\end{lem}
\begin{proof}
For the first assertion, we show that the condition for
weak co-synthesis holds against all possible behaviors of the 
channel between $\O$ and $\R$.
We have the following cases:
\begin{itemize}
\item {\em Case 1. Agents abort or resolve the protocol.\/}
In all traces where the agents abort or resolve the 
protocol, given $b$ is finite and that $\ttp'$ satisfies the
AGS constraints on the $\ttp$, by Lemma~\ref{lem-ttp-inviol}
(assertion 1), $\ttp'$ will eventually respond to the first and all 
subsequent requests such that $\varphi_\ttp$ is satisfied.
In all these traces, given the channels between the agents and the 
$\ttp$ are resilient, both agents get either the abort
token or non-repudiation evidences but never both.
This ensures $\varphi_\O$ and $\varphi_\R$ are satisfied.
\item {\em Case 2. The channel between $\O$ and $\R$ is resilient.\/}
In all traces where neither agent aborts nor resolves the protocol,
$\varphi_\ttp$ is satisfied trivially.
Further, the only refinements of the agents that neither abort nor 
resolve the protocol are those where $\set{m_1, m_3} \in \moves_{\O'}$
and $\set{m_2, m_4} \in \moves_{\R'}$.
Since $b$ is finite, the only choice of moves for $\O'$, since she
does not abort or resolve the protocol, are $m_1$ in state $v_0$ and 
$m_3$ in state $\set{M_1, \eor}$, after choosing at most $b$ idle moves 
at each state.
Similarly, the only choice of moves for $\R'$ are $\idle$ or
$m_2$ in state $\set{\eoo}$ and $\idle$ or $m_4$ in state 
$\set{\eoo, M_2, \siga^\O}$.
If $\R'$ never sends $m_2$, then $\O'$ will eventually abort 
the protocol after bounded idle time and this case reduces to
Case 1.
If $\R'$ never sends $m_4$, then $\O'$ will eventually resolve
the protocol after bounded idle time and this case reduces to
Case 1.
If $\R'$ sends $m_2$ and $m_4$ eventually, since the channel 
between $\O$ and $\R$ is assumed resilient,
messages $m_1$, $m_2$, $m_3$ and $m_4$ are eventually delivered
satisfying $\varphi_\O$ and $\varphi_\R$.
\item {\em Case 3. The channel between $\O$ and $\R$ is unreliable.\/}
Since $\O' \preceq \omax$, we have 
$a_1^\O \not\in \mov_{\O'}(v_0)$ and 
$a_1^\O \not\in \mov_{\O'}(\set{M_1, \eor, M_3})$; $\O'$ can abort
the protocol in all other states.
Therefore, $\O'$ satisfies the AGS constraints on $\O$.
Since $b$ is finite and $\O'$ cannot resolve the protocol before 
initiating it, the only choice of moves for $\O'$ in state $v_0$
is to send $m_1$ eventually.
If the channel between $\O$ and $\R$ does not deliver either
messages $m_1$ or $m_2$, the only choice of moves for $\O'$ is
to abort the protocol.
If either messages $m_3$ or $m_4$ are not delivered, then 
the only choice of moves for $\O'$ is to resolve the 
protocol.
In both these cases, since $\O'$ chooses to abort or resolve the 
protocol, by Case 1 the result follows.
\end{itemize}
We conclude that irrespective of the behavior of the channel 
between $\O$ and $\R$, if $b$ is finite, we have 
$\semb{\O' \para \R' \para \ttp' \para \Sc}
\subseteq (\varphi_\O \wedge \varphi_\R \wedge \varphi_\ttp)$.

For the second assertion, given an unbounded $b$, to show that weak 
co-synthesis fails, it suffices to show that there exists a behavior
of the agents, the $\ttp$ and the channels that
violates the condition for weak co-synthesis.
Consider a partial trace ending in protocol state 
$\set{M_1, \eoo, M_2, \eor, M_3, \siga^\O, \reqr^\O}$; messages
$m_1$, $m_2$ and $m_3$ have been received, $\R'$ chooses to
go idle, never sending $m_4$ and $\O'$ has sent $r_1^\O$.
Since $b$ is unbounded, if $\ttp'$ chooses to remain idle forever, 
then $\varphi_\O$ and $\varphi_\ttp$ are violated leading to a 
violation of $(\varphi_\O \wedge \varphi_\R \wedge \varphi_\ttp)$.
Therefore, given an unbounded $b$, we have
$\semb{\omax \para \rmax \para \tmax \para Sc}
\not\subseteq (\varphi_\O \wedge \varphi_\R \wedge \varphi_\ttp)$.
\qed
\end{proof}

From Lemma~\ref{lem-bounded-idle}, it is both necessary and sufficient
that the refinements $P' \preceq \pmax$ that satisfy the AGS 
constraints on the $\ttp$, also satisfy bounded idle time to ensure 
weak co-synthesis.
While $\O$ and the $\ttp$ should satisfy bounded idle time, there are
no restrictions on $\R$.
Using Lemma~\ref{lem-ttp-inviol}, Lemma~\ref{lem-max-r}, 
Lemma~\ref{lem-max-or} and Lemma~\ref{lem-bounded-idle} we now
present a proof of Lemma~\ref{lem-alt}.

\begin{proof}
{\bf (Proof of Lemma~\ref{lem-alt}).}
In one direction, consider an arbitrary refinement 
$P' = (\O', \R', \ttp') \in \pb$.
We show that the conditions of assume-guarantee synthesis are
satisfied as follows:
\begin{itemize}
\item {\em The implication condition for $\O$.\/}
Since $P' \preceq \pmax$, we have
$\O' \preceq \omax$, $\R' \preceq \rmax$ and $\ttp' \preceq \tmax$.
As $a_1^\O \not\in \mov_{\omax}(v_0)$ and
$a_1^\O \not\in \mov_{\omax}(M_1, \eor, M_3)$, the refinement
$P'$ satisfies the AGS constraints on $\O$.
Therefore, by Lemma~\ref{lem-max-or} (assertion 1),
we have $\semb{\O' \para \R \para \ttp \para \Sc} \subseteq
(\varphi_\R \wedge \varphi_\ttp) \im \varphi_\O$.
\item {\em The implication condition for $\R$.\/}
By Lemma~\ref{lem-max-r}, we have
$\semb{\O \para \R' \para \ttp \para \Sc} \subseteq
(\varphi_\O \wedge \varphi_\ttp) \im \varphi_\R$.
\item {\em The implication condition for the $\ttp$.\/}
Since $\ttp' \preceq \tmax$ and $\ttp'$ satisfies the AGS
constraints on the $\ttp$, by Lemma~\ref{lem-ttp-inviol} (assertion 1),
$\varphi_\ttp$ is satisfied irrespective of the behavior of $\O$
and $\R$, which implies 
$\semb{\O \para \R \para \ttp' \para \Sc} \subseteq
(\varphi_\O \wedge \varphi_\R) \im \varphi_\ttp$.
\item {\em The weak co-synthesis condition.\/}
Given $P'$ satisfies bounded idle time, by 
Lemma~\ref{lem-bounded-idle} we have
$\semb{\O' \para \R' \para \ttp' \para \Sc} \subseteq
(\varphi_\O \wedge \varphi_\R \wedge \varphi_\ttp)$; weak 
co-synthesis holds.
\end{itemize}
Since we have shown that the refinement $P'$ satisfies all
the implication conditions and the weak co-synthesis condition
of assume-guarantee synthesis, we have $P' \in P_{AGS}$.
Hence $\pb \subseteq P_{AGS}$.

In the other direction, consider an arbitrary refinement 
$P'' = (\O'', \R'', \ttp'') \in P_{AGS}$.
We show that $P'' \in \pb$ as follows:
\begin{itemize}
\item {\em The AGS constraints on $\O$.\/}
By Lemma~\ref{lem-max-or}, since it is both necessary and sufficient
that a refinement satisfy the AGS constraints on $\O$ to ensure the
implication condition $(\varphi_\R \wedge \varphi_\ttp) \im \varphi_\O$
is satisfied, given the implication condition holds, we conclude that 
$P''$ satisfies the AGS constraints on $\O$.
Therefore, $\O'' \preceq \omax$.
\item {\em The AGS constraints on the $\ttp$.\/}
By Lemma~\ref{lem-ttp-inviol}, since it is both necessary and sufficient
that a refinement satisfy the AGS constraints on the $\ttp$ to ensure the
implication condition $(\varphi_\O \wedge \varphi_\R) \im \varphi_\ttp$
is satisfied, given the implication condition holds, we conclude that 
$P''$ satisfies the AGS constraints on the $\ttp$ and 
$\ttp'' \preceq \tmax$.
\item {\em The bounded idle time condition.\/}
By Lemma~\ref{lem-bounded-idle}, since it is both necessary and sufficient
that a refinement satisfy bounded idle time to ensure weak co-synthesis,
since weak co-synthesis holds in this case, we conclude that $P''$ 
satisfies bounded idle time.
\item {\em $P'' \preceq \pmax$.\/}
Since we have shown that $\O'' \preceq \omax$ and $\ttp'' \preceq \tmax$,
we have $P'' \preceq \pmax$.
\item {\em $\pmin \preceq P''$.\/}
Since $\pmin$ is the smallest refinement in the set $P_{AGS}$, given
$P'' \in P_{AGS}$, it must be the case that $\pmin \preceq P''$.
\end{itemize}
For $P'' \in P_{AGS}$, as we have shown that $\pmin \preceq P'' \preceq
\pmax$, $P''$ satisfies the AGS constraints on the $\ttp$ and 
satisfies bounded idle time.
Thus we have $P'' \in \pb$ and hence $P_{AGS} \subseteq \pb$.
The result follows.
\qed
\end{proof}

We now present a proof of Lemma~\ref{psm-is-ags}.
We recall the {\em enhanced AGS constraints on the $\ttp$\/} 
below:
\begin{enumerate}
\item {\em Abort constraint.\/}
If the first request received by the $\ttp$ is $a_1^\O$, then her 
response to that request should be $\comp{a_2^\O, a_2^\R}$;
If the first request received by the $\ttp$ is $a_1^\R$, then her
response to that request should be $req^\O$;
\item {\em Resolve constraint.\/}
If the first request received by the $\ttp$ is a resolve
request, then her response to that request should be 
$\comp{r_2^\O, r_2^\R}$;
If the $\ttp$ receives $res^\O$ in response to $req^\O$ within
bounded idle time, then her response should be $\comp{r_2^\O, r_2^\R}$,
otherwise it should be $\comp{a_2^\O, a_2^\R}$.
\item {\em Accountability constraint.\/}
If the first response from the $\ttp$ is 
$\comp{x, y}$ or the first response from the $\ttp$ is $req^\O$ and
the next response is $\comp{x, y}$, then for all subsequent abort or 
resolve requests her response should be in the set 
$\set{\idle, x, y, \comp{x, y}}$.
\end{enumerate}

\begin{proof}
{\bf (Proof of Lemma~\ref{psm-is-ags}).}
From Protocol~\ref{prot:symm-main}, since the refinement $\osm$ 
does not abort the protocol either in the initial state $v_0$ or
after sending message $m_3$, we have $\osm$ satisfies
the AGS constraints on $\O$.
By our definition of the behavior of $\tsm$, we have $\tsm$
satisfies the enhanced AGS constraints on the $\ttp$.
From the definition of the main protocol in 
Protocol~\ref{prot:symm-main} and the abort subprotocol in
Protocol~\ref{prot:symm-abort}, since the resolve subprotocol is
identical to the KM protocol, we have $\osm$ and $\tsm$
satisfy the bounded idle time requirement.
We take $A = \set{\O, \R, \ttp}$ and show that there is no $Y$-attack 
for $Y \subseteq \set{\O, \R}$ through the following cases:
\begin{itemize}
\item {\em Case 1. $|Y| = 2$.\/}
In this case $Y = \set{\O, \R}$.
We show that 
$\semb{\O \para \R \para \tsm \para \Sc} \subseteq \varphi_\ttp$.
For all traces in $\semb{\O \para \R \para \tsm \para \Sc}$
where $\R$ does not abort the protocol, since $\tsm$ satisfies the 
enhanced AGS constraints on the $\ttp$, by
Lemma~\ref{lem-ttp-inviol} (assertion 1), $\varphi_\ttp$ is satisfied.
For all traces where $\R$ sends an abort request, the $\ttp$
sends $req^\O$.
If $\O$ responds with $res^\O$ within bounded idle time, then the 
$\ttp$ resolves the protocol for both $\O$ and $\R$ such that the 
AGS constraints on the $\ttp$ are satisfied.
If $\O$ does not send $res^\O$ within bounded idle time, then the 
$\ttp$ aborts the protocol, such that the AGS constraints on the 
$\ttp$ are satisfied.
For all subsequent abort requests from $\R$, the $\ttp$
response satisfies the AGS constraints on the $\ttp$.
All traces therefore satisfy $\varphi_\ttp$.
Hence, there is no $Y$-attack in this case.
\item {\em Case 2. $|Y| = 1$.\/}
In this case, either $Y = \set{\O}$ or $Y = \set{\R}$.
We have the following cases towards the proof:
\begin{itemize}
\item {\em Case (a). $Y = \set{\O}$.\/}
We show that 
$\semb{\O \para \rsm \para \ttp \para \Sc} \subseteq 
(\varphi_\O \wedge \varphi_\ttp) \im \varphi_\R$; it will
follow that
$\semb{\O \para \rsm \para \tsm \para \Sc} \subseteq
(\varphi_\O \wedge \varphi_\ttp) \im \varphi_\R$.
Consider the set of traces in 
$\semb{\O \para \rsm \para \ttp \para \Sc}$.
For all traces where $\R$ does not abort the protocol, by
Lemma~\ref{lem-max-r}, we have $\varphi_\O \wedge \varphi_\ttp
\im \varphi_\R$.
For all traces where $\R$ aborts the protocol, if he has received
$m_3$, then $\varphi_\R$ is satisfied.
For all traces where $\R$ aborts the protocol and message $m_3$
has not been received, if $\varphi_\ttp$ is violated, then
the implication holds and if $\varphi_\ttp$ is satisfied, then
either both agents get abort tokens or their respective
non-repudiation evidences, thus satisfying $\varphi_\R$.
We have shown that all traces satisfy the implication condition
$\varphi_\O \wedge \varphi_\ttp \im \varphi_\R$.
Since we have a fixed $\ttp$ that satisfies the AGS constraints
on the $\ttp$, we have $\varphi_\ttp$ is satisfied in all traces
by Case 1. 
As $\varphi_\O$ is satisfied by assumption, we conclude
$\varphi_\R$ is satisfied as well.
Therefore, there is no $Y$-attack in this case.
\item {\em Case (b). $Y = \set{\R}$.\/}
It can be shown that
$\semb{\osm \para \R \para \ttp \para \Sc} \not\subseteq
(\varphi_\R \wedge \varphi_\ttp) \im \varphi_\O$.
We show that, by fixing the $\ttp$, we have
$\semb{\osm \para \R \para \tsm \para \Sc} \subseteq
(\varphi_\R \wedge \varphi_\ttp) \im \varphi_\O$.
Consider the set of traces
$\semb{\osm \para \R \para \tsm \para \Sc}$.
For all traces where $\R$ does not abort the protocol, 
since $\O$ satisfies the AGS constraints on $\O$, 
by Lemma~\ref{lem-max-or}, we have $(\varphi_\R \wedge \varphi_\ttp)
\im \varphi_\O$.
If $\R$ aborts the protocol, since the $\ttp$ satisfies the
enhanced AGS constraints on the $\ttp$, and the channel between
$\O$ and the $\ttp$ is operational, $req^\O$ must have
been received by $\O$.
At this stage, if $\osm$ has sent message $m_3$, then the only choice 
of moves for $\osm$ to satisfy $\varphi_\O$ is $res^\O$; a request to 
resolve the protocol.
Since the channels are operational, there exists a bound on the idle 
time of the $\ttp$ such that both $req^\O$ and $res^\O$ can be delivered 
within this bound.
Moreover, as $\tsm$ satisfies the enhanced AGS constraints on the $\ttp$,
both $\O$ and $\R$ will be issued non-repudiation evidences and never
abort tokens, thus satisfying $\varphi_\O$.
If $\osm$ has not sent message $m_3$, then the only choice of moves
for $\osm$ to satisfy $\varphi_\O$ are $\idle$ or $res^\O$.
In all these traces, since $\tsm$ satisfies bounded idle time and the 
AGS constraints on the $\ttp$, either both agents get non-repudiation 
evidences or abort tokens but never both, thus satisfying $\varphi_\O$.
Therefore, all these traces satisfy
$(\varphi_\R \wedge \varphi_\ttp) \im \varphi_\O$, which given
$\varphi_\R$ is satisfied by assumption and $\varphi_\ttp$ is
satisfied by Case 1, implies $\varphi_\O$ is satisfied as well.
There is no $Y$-attack in this case.
\end{itemize}
\item {\em Case 2. $|Y| = 0$.\/}
In this case $Y = \emptyset$ and $(A \setm Y)' = \set{\osm, \rsm, \tsm}$.
Since $\psm$ satisfies bounded idle time, in all traces
where $\R$ does not abort the protocol, by Lemma~\ref{lem-bounded-idle},
the condition for weak co-synthesis is satisfied.
In all traces where $\R$ aborts the protocol, as $\tsm$ satisfies the 
AGS constraints on the $\ttp$, she sends $req^\O$.
In all these traces, since $\tsm$ and $\osm$ satisfy bounded idle time, 
and the channels are operational, $\osm$ chooses $\idle$ or sends 
$res^\O$ and $\tsm$ responds with either abort tokens or non-repudiation 
evidences but not both, leading to the satisfaction of 
$\varphi_\O$ and $\varphi_\R$.
Since $\varphi_\ttp$ is satisfied by Case 1, all these traces satisfy 
$(\varphi_\O \wedge \varphi_\R \wedge \varphi_\ttp)$.
Therefore, there is no $Y$-attack in this case.
\end{itemize}
The result follows.
\qed
\end{proof}

\begin{proof}
{\bf (Proof of Theorem~\ref{theo-psm-nonrep}).}
By Lemma~\ref{psm-is-ags}, it follows that if the $\ttp$ does not 
change her behavior, then $\psm$ is attack-free.
Further, by the weak co-synthesis condition, we have
$\semb{\osm \para \rsm \para \tsm \para \Sc} \subseteq
(\varphi_\O \wedge \varphi_\R \wedge \varphi_\ttp)$ and hence
by Theorem~\ref{objectives-imply-ft},
we have $\semb{\osm \para \rsm \para \tsm \para \Sc} \subseteq
\varphi_f$.
Thus $\psm$ satisfies fairness.
Using PCS we ensure abuse-freeness.
Since $\semb{\osm \para \rsm \para \tsm \para \Sc} \cap
(\diam \nro \wedge \diam \nrr) \ne \emptyset$, the refinement $\psm$
enables an exchange of signatures and hence is an exchange protocol.
We conclude that if the $\ttp$ does not change her behavior, then
$\psm$ is an attack-free fair non-repudiation protocol.
\qed
\end{proof}